\newif\ifllncs  %
\newif\iffull   %

\fulltrue

\def\shownotes{1}

\ifnum\shownotes=1
\newcommand{\authnote}[2]{\textcolor{red}{\textsf{#1 }\textcolor{blue}{ #2}}}
\else
\newcommand{\authnote}[2]{}
\fi
\newcommand{\qnote}[1]{{\authnote{Qipeng}{#1}}}
\newcommand{\ynote}[1]{{\authnote{Yilei}{#1}}}

\ifllncs
  \documentclass{llncs}

\else
  \documentclass[letterpaper,11pt]{article}
  \usepackage[in]{fullpage}
\fi

\usepackage
[pdftex,pagebackref,colorlinks=true,pdfpagemode=UseNone,urlcolor=blue,linkcolor=blue,citecolor=blue,pdfstartview=FitH]{hyperref}

\usepackage{amsmath,amsfonts, amssymb}
\usepackage{graphicx}
\usepackage{amsthm}
\usepackage{algorithm}
\usepackage{algpseudocode}
\usepackage{tikz}
\usepackage{hyperref}
\usepackage[capitalise]{cleveref}
\usepackage{breakcites}

\iffull
\fi

\ifllncs\else
\newtheorem{theorem}{Theorem}
\newtheorem{lemma}[theorem]{Lemma}
\newtheorem{definition}[theorem]{Definition}

\newtheorem{remark}[theorem]{Remark}

\newtheorem{corollary}[theorem]{Corollary}
\newtheorem{proposition}[theorem]{Proposition}
\fi
\newtheorem{fact}[theorem]{Fact}

\newcommand{\bra}[1]{\langle#1|}
\newcommand{\ket}[1]{|#1\rangle}
\newcommand{\bk}[2]{\langle#1|#2\rangle}
\newcommand{\kb}[2]{ \ket{#1}\bra{#2} }

\newcommand{\diag}[1]{\mathrm{diag} \left( #1\right) }
\newcommand{\gs}[1]{\mathsf{GS} \left( #1\right) }
\newcommand{\ngs}[1]{\mathsf{NGS} \left( #1\right) }

\newcommand{\poly}{{\sf poly}}
\newcommand{\negl}{{\sf negl}}

\newcommand{\sis}{{\sf SIS}}

\newcommand{\sistwo}{{\sf SIS}^{\infty}}

\newcommand{\approxSVP}{\mathsf{approx.SVP}}
\newcommand{\EDCP}{\mathsf{EDCP}}
\newcommand{\LWE}{\mathsf{LWE}}
\newcommand{\QLWE}{\mathsf{S\ket{LWE}}}
\newcommand{\LWEstate}{\mathsf{C\ket{LWE}}}

\newcommand{\zo}{\{0,1\}}

\newcommand{\la}{\leftarrow}

\newcommand{\F}{\mathbb{F}}

\newcommand{\lowerrounding}[1]{\left \lfloor #1 \right \rfloor}
\newcommand{\innerprod}[2]{\left \langle {#1}, {#2} \right \rangle}

\newcommand{\mat}[1] { {#1} }		%

\newcommand{\ary}[1] { {#1} }		%
\newcommand{\pmat}[1]{\begin{pmatrix}#1\end{pmatrix}}

\newcommand{\C}{\mathbb{C}}
\newcommand{\N}{\mathbb{N}}

\newcommand{\R}{\mathbb{R}}
\newcommand{\set}[1]{ \left\{ #1 \right\}  }   %

\newcommand{\Z}{\mathbb{Z}}

\newcommand{\QFT}{\mathsf{QFT}}

\newcommand{\cB}{\mathcal{B}}

\begin{document}

\title{Quantum Algorithms for Variants of \\Average-Case Lattice Problems via Filtering} 

\iffull
\author{
Yilei Chen\thanks{Tsinghua University. \texttt{chenyilei@mail.tsinghua.edu.cn}. \texttt{chenyilei.ra@gmail.com}.}
\and Qipeng Liu\thanks{Simons Institute for the Theory of Computing. \texttt{qipengliu0@gmail.com}.}
\and Mark Zhandry\thanks{Princeton University and NTT Research. \texttt{mzhandry@gmail.com}.}
}
\else
\author{
}
\fi

\maketitle

\begin{abstract}
We show polynomial-time quantum algorithms for the following problems: 
\begin{enumerate}
\item Short integer solution (SIS) problem under the \emph{infinity} norm, where the public matrix is very wide, the modulus is a polynomially large prime, and the bound of infinity norm is set to be half of the modulus minus a constant. 
\item Learning with errors (LWE) problem given LWE-like quantum states with polynomially large moduli and certain error distributions, including bounded uniform distributions and Laplace distributions.
\item Extrapolated dihedral coset problem (EDCP) with certain parameters. 
\end{enumerate}
The SIS, LWE, and EDCP problems in their standard forms are as hard as solving lattice problems in the worst case. However, the variants that we can solve are not in the parameter regimes known to be as hard as solving worst-case lattice problems. Still, no classical or quantum polynomial-time algorithms were known for the variants of SIS and LWE we consider. For EDCP, our quantum algorithm slightly extends the result of Ivanyos et al. (2018).

Our algorithms for variants of SIS and EDCP use the existing quantum reductions from those problems to LWE, or more precisely, to the problem of solving LWE given LWE-like quantum states. Our main contribution is solving LWE given LWE-like quantum states with interesting parameters using a filtering technique. 
\end{abstract}

%

\iffull
\newpage
\tableofcontents

\newpage
\fi

\iffull
\setcounter{page}{1}
\else
\pagestyle{plain}
\fi

\section{Introduction}

Solving the shortest vector problem (SVP) over lattices has been a target for designing efficient quantum algorithms for decades. In the literature, solving approximate SVP for \emph{all} lattices has been (classically or quantumly) reduced to the following problems:

\begin{enumerate}
	\item The short integer solution (SIS) problem, classically, initially shown by Ajtai~\cite{DBLP:conf/stoc/Ajtai96}.
	\item The dihedral coset problem (DCP), quantumly, initially shown by Regev~\cite{DBLP:conf/focs/Regev02}.
	\item The learning with errors problem (LWE), quantumly, initially shown by Regev~\cite{DBLP:conf/stoc/Regev05}.
\end{enumerate}

Therefore, to show an efficient quantum algorithm for approximate SVP in the worst-case, it suffices to construct an efficient quantum algorithm for any one of those average-case problems. However, no polynomial (or even subexponential) time quantum algorithms are known for SIS or LWE. For DCP, a subexponential quantum algorithm is given by Kuperberg~\cite{DBLP:journals/siamcomp/Kuperberg05}. But the quantum reduction shown by Regev~\cite{DBLP:conf/focs/Regev02} requires the DCP algorithm to be noise-tolerant, while the algorithm of Kuperberg is not. Let us also mention that over the past few years, efficient quantum algorithms for SVP for ideal lattices in certain parameter regimes have been shown in~\cite{CGS14,DBLP:conf/stoc/EisentragerHK014,biasse2016efficient,DBLP:conf/eurocrypt/CramerDPR16,DBLP:conf/eurocrypt/CramerDW17}. Still, showing a polynomial (or even subexponential) time quantum algorithm for SVP with polynomial approximation factors for \emph{all} lattices is widely open. 

The SIS and LWE problems are powerful tools for building cryptosystems, thus understanding the quantum hardness of those two problems is interesting in its own right. The SIS problem is typically used in constructing elementary cryptographic primitives such as one-way functions~\cite{DBLP:conf/stoc/Ajtai96}, collision-resistant hash functions~\cite{Goldreich-Goldwasser-Halevi96} 
digital signatures~\cite{DBLP:conf/stoc/GentryPV08}. The LWE problem is extremely versatile, yielding public-key cryptosystems~\cite{DBLP:conf/stoc/Regev05}, and advanced cryptographic capabilities such as fully homomorphic encryption (FHE)~\cite{DBLP:conf/focs/BrakerskiV11}, attribute-based encryption~\cite{DBLP:conf/stoc/GorbunovVW13}, and quantum FHE~\cite{DBLP:conf/focs/Mahadev18a}. The conjectured quantum hardness of SIS and LWE has also made lattice-based cryptosystems popular candidates for post-quantum cryptography standardization~\cite{DBLP:conf/africacrypt/DAnversKRV18,DBLP:conf/eurosp/BosDKLLSSSS18,DBLP:journals/tches/DucasKLLSSS18}.

\subsection{Background of SIS, LWE, DCP, and our main results}

We show polynomial-time quantum algorithms for certain variants of SIS, LWE, and DCP. 
Our quantum algorithms for the variants of SIS and DCP go through the existing quantum reductions from those problems to LWE, or more precisely, to the problems of \emph{Constructing quantum LWE states} ($\LWEstate$) and \emph{Solving LWE given LWE-like quantum states} ($\QLWE$). 
In fact, the heart of our results is showing a quantum filtering technique for solving those quantum versions of LWE. 

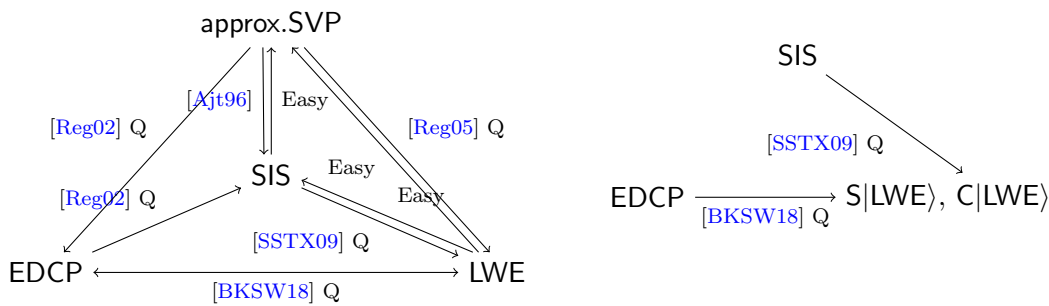
\begin{figure}
  \centering
    \begin{tikzpicture}
      \begin{scope}[xshift=1cm, yshift = 0.7cm]
        \draw node at (0  , 2  ) (C) {$\approxSVP$};
        \draw node at (0  , 0   ) (O) {$\sis$};
        \draw node at (-3 ,-1.3 ) (A) {$\EDCP$};
        \draw node at (3  ,-1.3 ) (B) {$\LWE$};

        \draw [<-] (O) -- node[auto,swap] {\scriptsize \cite{DBLP:conf/focs/Regev02} Q} (A);
        \draw [->] (C) -- node[auto,swap] {\scriptsize \cite{DBLP:conf/focs/Regev02} Q} (A);
        \draw [->] (B.140) -- node[auto,swap] {\scriptsize Easy} (O.350);
        \draw [->] (O) -- node[auto,swap] {\scriptsize \cite{DBLP:conf/asiacrypt/StehleSTX09} Q} (B);
        \draw [<->] (A) -- node[auto,swap] {\scriptsize \cite{DBLP:conf/pkc/BrakerskiKSW18} Q} (B);
        \draw [->] (O) -- node[auto,swap] {\scriptsize Easy} (C);
        \draw [->] (C.250) -- node[auto,swap] {\scriptsize \cite{DBLP:conf/stoc/Ajtai96}} (O.110);
        \draw [<-] (C) -- node[auto,swap] {\scriptsize Easy} (B);
        \draw [<-] (B.110) -- node[auto,swap] {\scriptsize \cite{DBLP:conf/stoc/Regev05} Q} (C.325);
      \end{scope}

      \begin{scope}[xshift=8cm, yshift = 1.7cm]
        \draw node at (0  , 0.6   ) (O) {$\sis$};
        \draw node at (-2 ,-1.3 ) (A) {$\EDCP$};
        \draw node at (2 ,-1.3 ) (B) {$\QLWE$, $\LWEstate$};

        \draw [->] (O) -- node[auto,swap] {\scriptsize \cite{DBLP:conf/asiacrypt/StehleSTX09} Q} (B.60);
        \draw [->] (A) -- node[auto,swap] {\scriptsize \cite{DBLP:conf/pkc/BrakerskiKSW18} Q} (B);
      \end{scope}
    \end{tikzpicture}
    \caption{Left: An overview of the reductions between SVP, SIS, EDCP, and LWE. ``$A\to B$'' means Problem $A$ reduces to Problem $B$. ``Q'' means quantum. Right: The reductions used in our quantum algorithms.   }
  \label{fig:reductions}
\end{figure}

Let us now provide more background of SIS, LWE, and DCP, then state our main results.

\subsubsection{SIS } 
Let us first recall the standard definition of the SIS problem.
\begin{definition}[Short integer solution (SIS) problem~\cite{DBLP:conf/stoc/Ajtai96}]
	Let $n, m, q$ be integers such that $m = \Omega(n\log q)\subseteq\poly(n)$. Let $\beta$ be a positive real number such that $\beta<q$. Let $A$ be a uniformly random matrix over $\Z_q^{n \times m}$. The SIS problem $\sis_{n, m, q, \beta}$ asks to find a nonzero vector $x \in \mathbb{Z}^m$ such that $\| x \|_2 \leq \beta$ and $A x \equiv 0 \pmod q$.
\end{definition}

The SIS problem is shown to be as hard as solving approximate SVP for all lattices~\cite{DBLP:conf/stoc/Ajtai96}. The reductions are improved via a series of works~\cite{DBLP:conf/focs/CaiN97,DBLP:conf/stoc/Micciancio02,MicciancioRegev07,DBLP:conf/stoc/GentryPV08,micciancio2013hardness}. %
Several variants of the SIS problem are studied in the literature. The most common variant is the one that changes the restriction of the solution. The solution is bounded in $\ell_p$ norm for some $p\geq 0$, or even the $\ell_\infty$ norm, instead of bounded in $\ell_2$ norm. In this paper, we look at the variant where the solution is bounded by its $\ell_\infty$ norm. More precisely, we use $\sistwo_{n,m,q,\beta}$ to denote the variant of SIS where the solution $x$ is required to satisfy $\|x\|_\infty\leq \beta$. When $\beta = 1$, it corresponds to the subset-sum problem where the solution is bounded in $\set{-1, 0, 1}$.

Bounding the SIS solution in its $\ell_\infty$ norm is used quite commonly in cryptography due to its simplicity (it is used, e.g., in~\cite{brakerski2015constrained}). When the parameters are set so that $\beta\sqrt{m}>q$, i.e., when $m$ is relatively large compared to $q/\beta$, we are not aware of any worst-case problem that is reducible to $\sistwo_{n,m,q,\beta}$. Still, such parameter settings are used in cryptosystems. In a recent practical signature scheme proposed by Ducas et al.~\cite{DBLP:journals/tches/DucasKLLSSS18}, the security of the scheme relies on (the ``Module'' version of) $\sistwo_{n,m,q,\beta}$ with $\beta\sqrt{m}>q$. In their security analysis, the authors mention that the problem of $\sistwo$ by itself has not been studied in-depth. Most of the algorithms they can think of for $\sistwo$ are the ones designed for solving SIS or SVP in the $\ell_2$ norm, such as BKZ~\cite{schnorr1994lattice}.

To date, the only algorithm we are aware of that takes advantage of the $\ell_\infty$-norm bound has the following features. It solves $\sistwo_{n,m,q,\beta}$ with a highly composite $q$ and a very large $m$. For example, it is a polynomial-time algorithm for $\sistwo_{n,O(n^c),2^c,1}$ when $c$ is a constant. The algorithm is classical, folklore, and we include a formal description of the algorithm in \Cref{sec:SIS_infinite_alg}. It was not clear how to solve $\sistwo_{n,m,q,\beta}$ when $q$ is a polynomial prime and $\beta$ is just slightly smaller than $q/2$, even if $m$ is allowed to be an arbitrary polynomial.

We show a polynomial-time quantum algorithm for $\sistwo_{n,m,q,\beta}$ where $q$ is a polynomial prime modulus, $\beta = \frac{q-c}{2}$ for some constant $c$, and $m$ is a large polynomial. 
\begin{theorem}\label{thm:SISwithpolymod_intro}
Let $c>0$ be a constant integer, $q>c$ be a polynomially large prime modulus. Let $m\in \Omega\left( (q-c)^{3} \cdot n^{{c + 1}}  \cdot q\cdot \log q \right)\subseteq \poly(n)$, there is a polynomial-time quantum algorithm that solves $\sistwo_{n,m,q,\frac{q-c}{2}}$. 
\end{theorem}
\begin{remark}
Note that if $\beta=q/2$, then a solution can be found classically by simply solving $Ax\equiv 0\pmod q$ over $\mathbb{Z}_q$ using Gaussian elimination. Then for each entry in $x$, pick the representative over $\Z$ that lies in the range $[-q/2,q/2)$. This classical algorithm also extends to $\beta=\frac{q-c}{2}$ when $q=\Omega(n)$. In particular, as long as all the entries of $x$ are at least $c/2$ far from $q/2$, $x$ will be a valid solution. In the regime $q=\Omega(n)$, a random solution to $Ax\equiv 0 \pmod q$ will satisfy this with probability at least $O((1-c/n)^n)=O(e^{-c})$, a constant.
Theorem~\ref{thm:SISwithpolymod_intro} thus gives a non-trivial algorithm for $\sistwo_{n,m,q,\frac{q-c}{2}}$ when $q\in o(n)$, for which (to the best of our knowledge), no prior classical or quantum algorithm was known. 
\end{remark}

\begin{remark}
Our algorithm can also solve a variant of SIS where the each entry of the solution is required to be in an arbitrary subset $S$ of $\Z_q$ such that $q - |S| = c$, where $c$ is a constant (instead of the subset $[-\beta, \beta]\cap\Z$ of $\Z_q$). The width of the $A$ matrix is required to satisfy $m\in \Omega\left( (q-c)^{3} \cdot n^{{c + 1}} \cdot q\cdot \log q \right)\subseteq \poly(n)$. For example, suppose $q = 3$ and $m \in\Omega(n^2)$, our algorithm is able to provide a $\set{0, 1}$-solution for SIS. 
\end{remark}

Let us remark that our algorithm does not improve upon the existing algorithms for breaking the signature scheme in~\cite{DBLP:journals/tches/DucasKLLSSS18} since we require $m$ to be very large, while the $m$ used in~\cite{DBLP:journals/tches/DucasKLLSSS18} is fairly small.

\subsubsection{LWE } 
Let us first recall the classical definition of the LWE problem.
\begin{definition}[Learning with errors (LWE)~\cite{DBLP:conf/stoc/Regev05}]
Let $n$, $m$, $q$ be positive integers. 
Let $u \in \mathbb{Z}_q^n$ be a secret vector. 
The learning with errors problem $\LWE_{n,m,q,\mathcal{D}_{\sf noise}}$ asks to find the secret vector $u$ given access to an oracle that outputs $a_i$, $a_i \cdot u + e_i \pmod q$ on its $i^{th}$ query, for $i = 1, ..., m$. Here each $a_i$ is a uniformly random vector in $\mathbb{Z}_q^n$, and each error term $e_i$ is sampled from a distribution $\mathcal{D}_{\sf noise}$ over $\Z_q$. 
\end{definition}

Regev~\cite{DBLP:conf/stoc/Regev05} shows if there is a polynomial-time algorithm that solves $\LWE_{n,m,q,\mathcal{D}_{\sf noise}}$ where $\mathcal{D}_{\sf noise}$ is Gaussian and $m$ can be an arbitrary polynomial, then there is a quantum algorithm that solves worst-case approximate SVP.
Note that in Regev's definition, the LWE samples are completely classical. In the variants of LWE we consider,  the error distribution appears in the amplitude of some quantum states. Those quantum variants of LWE were implicitly used in the quantum reductions in~\cite{DBLP:conf/asiacrypt/StehleSTX09,DBLP:conf/pkc/BrakerskiKSW18}, but they have not been made formal. Looking ahead, our new quantum algorithms make explicit use of the quantum nature of the noise distribution. 

Our quantum algorithm for $\sistwo$ adapts the quantum reduction from $\sis$ to the problem of \emph{constructing LWE states} implicitly used in~\cite{DBLP:conf/asiacrypt/StehleSTX09}.

\begin{definition}\label{def:LWEstateproblem}
Let $n$, $m$, $q$ be positive integers. Let  $f$ be a function from $\Z_q$ to $\R$. 
The problem of constructing LWE states $\LWEstate_{n,m,q,f}$ asks to construct a  quantum state of the form 
$\sum_{u\in\Z_q^n} \bigotimes_{i = 1}^m\left( \sum_{e_i\in\Z_q} f(e_i) \ket{ a_i \cdot u + e_i \bmod q}\right)$, given the input $\set{ a_i }_{i = 1, ..., m}$ where each $a_i$ is a uniformly random vector in $\mathbb{Z}_q^n$.
\end{definition}

Our quantum algorithm for $\EDCP$ adapts the quantum reduction from $\EDCP$ to the problem of \emph{solving LWE given LWE-like quantum states} implicitly used in~\cite{DBLP:conf/pkc/BrakerskiKSW18}.

\begin{definition}\label{def:QLWEproblem}
Let $n$, $m$, $q$ be positive integers. Let $f$ be a function from $\Z_q$ to $\R$.
Let $u \in \mathbb{Z}_q^n$ be a secret vector. 
The problem of solving LWE given LWE-like states $\QLWE_{n,m,q,f}$ asks to find $u$ given access to an oracle that outputs independent samples $a_i$, $\sum_{e_i\in\Z_q} f(e_i) \ket{ a_i \cdot u + e_i \bmod q}$ on its $i^{th}$ query, for $i = 1, ..., m$. Here each $a_i$ is a uniformly random vector in $\mathbb{Z}_q^n$. 
\end{definition}

We would like to remark that in the problem $\LWEstate$, there is no secret vector $u$; the goal is to construct an equal superposition of all LWE states for all possible $u$. Whereas for the problem $\QLWE$, the goal is to find the secret vector $u$ given samples of LWE states for this particular secret vector. 

Let us briefly discuss the relations among $\LWE$, $\QLWE$, and $\LWEstate$. 
If we set $f$ as $\sqrt{ \mathcal{D}_{\sf noise}}$, then an efficient algorithm for solving $\LWE_{n,m,q,\mathcal{D}_{\sf noise}}$ implies efficient algorithms for solving $\LWEstate_{n,m,q,f}$ and $\QLWE_{n,m,q,f}$. However, solving $\LWEstate_{n,m,q,f}$ or $\QLWE_{n,m,q,f}$ does not necessarily imply solving $\LWE_{n,m,q,\mathcal{D}_{\sf noise}}$ in general. An algorithm for solving  $\LWEstate_{n,m,q,f}$ only implies an efficient algorithm for solving $\LWE_{n,m,q,\mathcal{D}_{\sf noise}}$ when $m$ is small compared to the ratio of the ``widths'' of $f$ and $\mathcal{D}_{\sf noise}$; we will explain in details in~\S\ref{sec:intro:open}. 

Let us also remark that the $\LWEstate$ and $\QLWE$ problems 
we define are different from the problem of ``LWE with quantum samples'' defined in~\cite{grilo2019learning}. In their definition, the quantum LWE samples are of the form 
$\sum_{a\in\Z_q^n}\ket{a}\ket{a\cdot u+e}$, where the error $e$ is classical and $a$ is in the quantum state. This variant of quantum LWE is easy to solve~\cite{grilo2019learning}, but the idea in the algorithm does not carry to the quantum LWE variants we are interested in.

In~\cite{DBLP:conf/stoc/Regev05} (followed by~\cite{DBLP:conf/asiacrypt/StehleSTX09,DBLP:conf/pkc/BrakerskiKSW18} and most of the papers that use LWE), the noise distribution $\mathcal{D}_{\sf noise}$ or $f$ is chosen to be Gaussian. One of the nice features of a Gaussian function $f$ is that both $f$ and its discrete Fourier transform (DFT) %
(over $\Z_q$), defined as 
\[ \hat{f}: \Z_q \to \C, ~~~ 
   \hat{f}: y\mapsto \sum_{x\in\Z_q} \frac{1}{\sqrt{q}}\cdot e^{\frac{2 \pi i xy}{q}} \cdot f(x),   \] 
are negligible beyond their centers. Such a feature of $\hat{f}$ is crucial in establishing the quantum reductions among lattice problems in~\cite{DBLP:conf/stoc/Regev05,DBLP:conf/asiacrypt/StehleSTX09,DBLP:conf/pkc/BrakerskiKSW18}.

Other choices of noise distribution are also used for LWE in the literature. One popular option is to let $f$ be the bounded uniform distribution over $[-B, B]$ for some $0<B<\frac{q}{2}$. For certain choices of $n,m,q,B$, (classical) LWE with $B$-bounded uniform error is proven to be as hard as LWE with Gaussian noise~\cite{DBLP:conf/eurocrypt/DottlingM13}. 
On the other hand, Arora and Ge~\cite{DBLP:conf/icalp/AroraG11} present a classical algorithm for breaking LWE with a prime modulus $q$ when the support $S$ of the LWE error distribution is very small. It requires $m \in \Omega(n^{|S|})$ and runs in time $\poly(n^{|S|})$. When $B\in\omega(1)$ and $q$ is a prime, no polynomial-time quantum algorithm has been published for $\LWE$, $\LWEstate$, or $\QLWE$.

We show when the noise distribution $f$ is chosen such that $\hat{f}$ is \emph{non-negligible} over $\Z_q$, then we can solve both $\LWEstate$ and $\QLWE$ in quantum polynomial-time.
\begin{theorem}\label{thm:solvingLWEstate_intro}
Let $n\in\N$ and $q\in\poly(n)$. 
Let $f:\Z_q\to \R$ be the amplitude for the error state such that the state $\sum_{e\in\Z_q} f(e)\ket{e}$ is efficiently constructible and $\eta:=\min_{y\in\Z_q} |\hat{f}(y)|$ is non-negligible.
Let $m\in \Omega\left( n  \cdot q / \eta^2 \right)\subseteq \poly(n)$, there exist polynomial-time quantum algorithms that solve $\LWEstate_{n,m,q,f}$ and $\QLWE_{n,m,q,f}$.
\end{theorem}

Although the theorem does not cover the case where $f$ is Gaussian, it does cover some interesting error distributions $f$, such as when $f$ is super-Gaussian (i.e., when $f(x) = e^{-|x/B|^p}$, for $0<p<2$, $0<B<q$). It also covers the case where $f$ is the bounded uniform distribution. The following is a corollary of Theorem~\ref{thm:solvingLWEstate_intro} given that the DFT of bounded uniform distribution is non-negligible over $\Z_q$.

\begin{corollary}\label{thm:solvingLWEstate_Buniform_intro}
Let $n\in\N$ and $q\in\poly(n)$. Let $B\in\Z$ such that $0<2B+1<q$ and $\gcd(2B+1, q) = 1$. Let $f:\Z_q\to \R$ be $f(x) := 1/\sqrt{2B+1}$ when $x\in [-B, B]\cap\Z$ and $0$ elsewhere.
Let $m\in \Omega\left( n  \cdot q^4 \cdot(2B+1) \right)\subseteq \poly(n)$, there exist polynomial-time quantum algorithms that solve $\LWEstate_{n,m,q,f}$ and $\QLWE_{n,m,q,f}$.
\end{corollary}

Our quantum algorithms for $\sistwo$ and $\EDCP$ (i.e., Theorem~\ref{thm:SISwithpolymod_intro} and Theorem~\ref{thm:DCPwithpolymod_intro}) are obtained from the following variant of Theorem~\ref{thm:solvingLWEstate_intro}, where the noise amplitude for the quantum LWE problems is set to be the DFT of the bounded uniform distribution. 
\begin{theorem}\label{thm:theLWEstatementneeded_intro}
Let $q$ be a polynomially large prime modulus. Let $B\in\Z$ be such that $q-(2B+1) = c$ is a constant. 
Let $f:Z_q\to \R$ be the bounded uniform distribution over $[-B, B]\cap\Z$. 
Let $m\in \Omega\left( (q-c)^{3} \cdot n^{{c + 1}} \cdot q\cdot \log q \right)\subseteq \poly(n)$.
There exist polynomial-time quantum algorithms that solve $\LWEstate_{n,m,q,\hat f}$ and $\QLWE_{n,m,q,\hat f}$. 
\end{theorem}

\subsubsection{DCP}\label{sec:EDCPintro}
Let us introduce the variant of DCP defined by Brakerski et al.~\cite{DBLP:conf/pkc/BrakerskiKSW18}.
\begin{definition}[Extrapolated Dihedral Coset Problem (EDCP)]\label{def:EDCP}
Let $n\in\N$ be the dimension, $q\geq 2$ be the modulus, and a function $D:\Z_q\to \R$, consists of $m$ input states of the form
\[ \sum_{j\in\Z_q} D(j) \ket{j} \ket{ x+j\cdot s}, \]
where $x \in \Z_q^n$ is arbitrary and $s\in\Z_q^n$ is fixed for all $m$ states. We say that an algorithm solves $\EDCP_{n,m,q,D}$ if it outputs $s$ with probability $\poly(1/(n \log q))$ in time $\poly(n \log q)$.
\end{definition}

In this paper we are interested in the parameter setting where $n$ is the security parameter and $q\in\poly(n)$. Although not strictly needed in this paper, let us briefly recall how the variants of the dihedral coset problem evolve. The original dihedral coset problem is a special case of EDCP where $n=1$, $q$ is exponentially large, and $D$ is the uniform distribution over $\zo$. 
Solving DCP implies solving the dihedral hidden subgroup problem. 
The two-point problem defined by Regev~\cite{DBLP:conf/focs/Regev02} is another special case of EDCP where $D$ is the uniform distribution over $\zo$, and $n$ is the security parameter. It was used as an intermediate step for establishing the reduction from approximate SVP to DCP.
When the distribution $D$ is non-zero beyond $\zo$, the EDCP problem does not necessarily correspond to any versions of the hidden subgroup problem. The reason that Brakerski et al.~\cite{DBLP:conf/pkc/BrakerskiKSW18} considers a distribution $D$ supported beyond $\zo$ is to establish a reduction from EDCP to LWE. Therefore, combining with the reduction from LWE to EDCP (by adapting Regev's reduction~\cite{DBLP:conf/focs/Regev02}), they show that EDCP, as a natural generalization of DCP, is equivalent to LWE. 

Efficient quantum algorithms are known for variants of EDCP when the modulus $q$ and the distribution $D$ satisfy certain conditions~\cite{DBLP:conf/stoc/FriedlIMSS03,DBLP:conf/soda/ChildsD07,DBLP:conf/esa/IvanyosPS18}. Let us remark that EDCP with those parameter settings are not known to be as hard as worse-case SVP or LWE through the reductions of~\cite{DBLP:conf/focs/Regev02} or~\cite{DBLP:conf/pkc/BrakerskiKSW18}.  

In this paper we show polynomial-time quantum algorithms that solve EDCP with the following parameter settings.

\begin{theorem}\label{thm:DCPwithpolymod_fg_intro}
Let $n\in\N$ and $q \in \poly(n)$. 
Let $f:\Z_q\to \R$ be such that the state $\sum_{e\in\Z_q} f(e)\ket{e}$ is efficiently constructible and $\eta:=\min_{z\in\Z_q}|\hat{f}(z)|$ is non-negligible. 
Let $m\in \Omega\left( n \cdot q / \eta^2 \right)\subseteq \poly(n)$. 
There is a polynomial time quantum algorithm that solves $\EDCP_{n,m,q,\hat f}$
\end{theorem}

\begin{theorem}\label{thm:DCPwithpolymod_intro}
Let $n\in\N$ and $q \in \poly(n)$. Let $c$ be a constant integer such that $0<c<q$. Let $m\in \Omega\left( (q-c)^{3} \cdot n^{{c + 1}}   \cdot q\cdot \log q \right)\subseteq \poly(n)$, there is a quantum algorithm running in time $\poly(n)$ that solves $\EDCP_{n,m,q,D}$ where $D$ is the uniform distribution over $[0,q-c)\cap\Z$. 
\end{theorem}

We note that EDCP with the parameters in Theorem~\ref{thm:DCPwithpolymod_intro} has already been solved in the work of Ivanyos et al.~\cite{DBLP:conf/esa/IvanyosPS18} by a quantum algorithm with similar complexity. The parameters in Theorem~\ref{thm:DCPwithpolymod_fg_intro} are not covered by the result in~\cite{DBLP:conf/esa/IvanyosPS18}, but the implication of such a parameter setting is unclear. Nevertheless, we include our result to demonstrate the wide applicability of our techniques. We will compare our algorithm with the one in \cite{DBLP:conf/esa/IvanyosPS18} in \S\ref{sec:related}.

\subsection{Solving the quantum versions of LWE via filtering}

As mentioned, our main technical contribution is to solve $\QLWE$ and $\LWEstate$ (the quantum versions of LWE we define) with interesting parameters using  a \emph{filtering} technique. Let us first explain the basic idea of filtering, then extend it to the general case.

\paragraph{The basic idea of filtering. }
To illustrate the basic idea of filtering, let us focus on how to use it to solve $\QLWE$, namely, learning the secret $u\in\Z_q^n$ given a uniformly random matrix $A\in\Z_q^{n\times m}$ and the following state:
\begin{equation}\label{eqn:inputstate}
\ket{\phi_u} := \bigotimes_{i = 1}^{m} \sum_{e_i\in\Z_q}  f(e_i) \ket{\; (u^T A)_i + e_i \pmod q\;}.   
\end{equation}
Let us remark that an efficient quantum algorithm for $\QLWE_{n,m,q,f}$ does not necessarily imply an efficient quantum algorithm for $\LWEstate_{n,m,q,f}$,
since the quantum algorithm for $\QLWE_{n,m,q,f}$ may, for example, destroy the input state. However, the quantum algorithm we show for $\QLWE_{n,m,q,f}$ directly works for $\LWEstate_{n,m,q,f}$, so we focus on $\QLWE_{n,m,q,f}$.

Let us assume $m$ can be an arbitrary polynomial of $n$, $q$ is a constant prime. The readers can think of $f$ as any distribution. For readers who would like to have a concrete example, you can think of $f$ as the QFT of bounded uniform distribution, i.e., let $g(z) := 1/\sqrt{2\beta+1}$ for $z\in[-\beta, \beta]\cap\Z$ and $0$ elsewhere, then set $f := \hat g$ ($f$ is then the discrete sinc function, but in the analysis we will not use the expression of $f$ at all, we will only use $g$). %
By solving $\QLWE_{n,m,q,f}$ and $\LWEstate_{n,m,q,f}$ with such a choice of $f$, we can get a polynomial quantum algorithm for $\sistwo_{n,m,q,\beta}$ with a constant prime $q$ and any $\beta\in [1, q/2)$, which was not known before. 
All the details of the analysis will be given in \S\ref{sec:idea_filtering}. 
Here let us explain the basic idea of filtering using this example.

Let us define %
\begin{equation*}
    \text{for }v\in\Z_q\text{,  }\ket{ \psi_v } := \sum_{e\in\Z_q} f(e) \ket{ (v + e) \pmod q }.
\end{equation*}
Therefore the input state in Eqn~\eqref{eqn:inputstate} can also be written as 
\begin{equation*}
     \ket{\phi_u} = \bigotimes_{i = 1}^{m}\ket{ \psi_{(u^T A)_i} }.
\end{equation*}
To learn $u$ from $\ket{\phi_u}$, our algorithm proceeds in two stages: first we look at each coordinate $\ket{ \psi_{(u^T A)_i} }$ for $i = 1, ..., m$ separately, with the goal of learning some classical information about each coordinate of $u^T A$. We then continue with a classical step, which uses the information obtained about each coordinate of $u^T A$ to learn $u$.

\paragraph{Warm-up 1: Orthogonal states.} Suppose the vectors in the set $\set{ \ket{\psi_v} }_{v\in\Z_q}$ were orthogonal. Then we could define a unitary $U$ such that $U\ket{\psi_v}=\ket{v}$. We could then apply this unitary component-wise to $\ket{\phi_u}$ and measure the results in the computational basis, learning $u^T A$. Gaussian elimination then recovers $u$.

\paragraph{Warm-up 2: Filtering out a single value.} Unfortunately, the $\ket{\psi_v}$ will typically not be orthogonal, so such a unitary as above will not exist. This means we cannot learn $v$ with certainty from $\ket{\psi_v}$.

Nevertheless, we can learn \emph{some} information about $v$ from $\ket{\psi_v}$. Concretely, pick some value $y\in\Z_q$, and consider an arbitrary unitary $U_y$ such that
\[
U_y\ket{\psi_y}=\ket{0}.
\]
Now imagine applying $U_y$ to $\ket{\psi_v}$, and measuring in the computational basis. If $v=y$, then the measurement will always give 0. Unfortunately,  since the $\ket{\psi_v}$ are not orthogonal, measuring $U_y\ket{\psi_v}$ for $v\neq y$ may also give 0. Therefore, while a 0 outcome gives us some prior on the value of $v$, it does not let us conclude anything for certain.

On the other hand, if a measurement gives a \emph{non-zero} value, then we know for certain that $v\neq y$. This is the basic idea of our filtering approach: we filter out the case where $v=y$, learning an inequality constraint on $v$. This can be seen as a weak form of unambiguous state discrimination~\cite{Peres1988HowTD}, where the measurement either gives unambiguous information about the unknowns or is thrown away. It turns out that, in some parameter regimes, learning such non-equality constraints will let us compute $u$.

Concretely, given an unknown state $\ket{\phi_u}$, we choose an independent random $y_i$ for each coordinate, apply the unitary $U_{y_i}$ to the $i$th coordinate, and measure. Any measurement result that gives 0, we throw away; for typical $\ket{\psi_v}$, few measurements will give 0. The remaining results yield inequality constraints of the form $(u^T A)_i\neq y_i$. We then apply the classical Arora-Ge algorithm~\cite{DBLP:conf/icalp/AroraG11} to these constraints. This algorithm works by viewing the inequality constraints as degree $q-1$ constraints and then re-linearizing them. This process converts the inequality constraints into equality constraints, but at the cost of blowing up the number of unknowns to $\approx n^{q-1}$. In the regime where $q$ is a constant and $m$ is a sufficiently large polynomial, the system can be solved in polynomial-time using Gaussian elimination.

\paragraph{Our algorithm: filtering out multiple values.} Our algorithm so far is limited to filtering out a single value, which in turn limits us to a constant $q$, due to our use of Arora-Ge.

In order to get a polynomial-time algorithm for larger $q$, we must filter out more points; in fact, in order to use Arora-Ge, we need our constraints to have constant degree, which in turn means we must filter out all but a constant number of elements of $\Z_q$. Filtering out so many points requires care.

Consider the goal of filtering out two values. If there exists, for $y_0,y_1\in\Z_q$, a unitary $U_{y_0,y_1}$ such that
\[U_{y_0,y_1}\ket{\psi_{y_b}}=\ket{b},\]
then we could apply $U_{y_0,y_1}$ and measure in the computational basis. If the result is not equal to 0 or 1, then we know that $v\notin\{y_0,y_1\}$, thus filtering out two values.

In general such a unitary does not exist, as it would require $\ket{\psi_{y_0}}$ and $\ket{\psi_{y_1}}$ to be orthogonal. Instead what we do is to define a unitary $U_{y_0,y_1}$ such that 
\[U_{y_0,y_1}\ket{\psi_{y_b}}\in\textrm{Span}(\ket{0},\ket{1})\enspace.\]
This method also naturally extends to filtering a larger number of $y$ values. The limitation is that, as the number of $y$ increases, the probability of getting a successful measurement (where ``success'' means, e.g. getting a result other than 0,1) decreases. For example, suppose the $\ket{\psi_v}$ all lie in the space of dimension $d\ll q$. Then after excluding $d$ values, the probability of a successful measurement will be 0. Even if the vectors are technically linear independent but close to a $d$-dimensional subspace, the probability will be non-zero but negligible. This, for example, rules out an algorithm for the case where $f$ is discrete Gaussian.

Therefore, whether or not the algorithm will succeed depends crucially on the ``shape'' of the states $\ket{\psi_v}$, and in particular, the distribution $f$. Our applications roughly follow the outline above, analyzing specific cases of $\ket{\psi_v}$. Our main observation is that, since all the vectors $\ket{\psi_v}$ are just shifts of a single fixed vector, we can construct a unitary operator by taking the normalized Gram-Schmidt orthogonalization of a circulant matrix $M_f$, defined by
\begin{equation}\label{eqn:M_f}
M_f:= [ \ket{\psi_v}, \ket{\psi_{v+1}}, ..., \ket{\psi_{v+q-1}} ]. 
\end{equation}
This allows us to relate the success probability of filtering out $q-1$ values to the length of the last Gram-Schmidt vector of $M_f$ (before normalization). The length of the last Gram-Schmidt vector is related to the eigenvalues of the circulant matrix $M_f$, and it can be bounded in terms of the discrete Fourier transform $\hat{f}$ of $f$. Our calculation suggests that when $\hat{f}$ is \emph{non-negligible} over all the values in $\Z_q$, the success probability of correctly guessing each coordinate is non-negligible. Therefore when $m$ is a sufficiently large polynomial, we get a polynomial-time algorithm for $\QLWE_{n,m,q,f}$. In Figure~\ref{fig:examplesf} we give four examples of error amplitudes. It shows that if the minimum of $\hat f$ is non-negligible, then the length of the last Gram-Schmidt of $M_f$ is non-negligible.

\iffull
\begin{figure}[ht]
\centering
\includegraphics[scale=0.24]{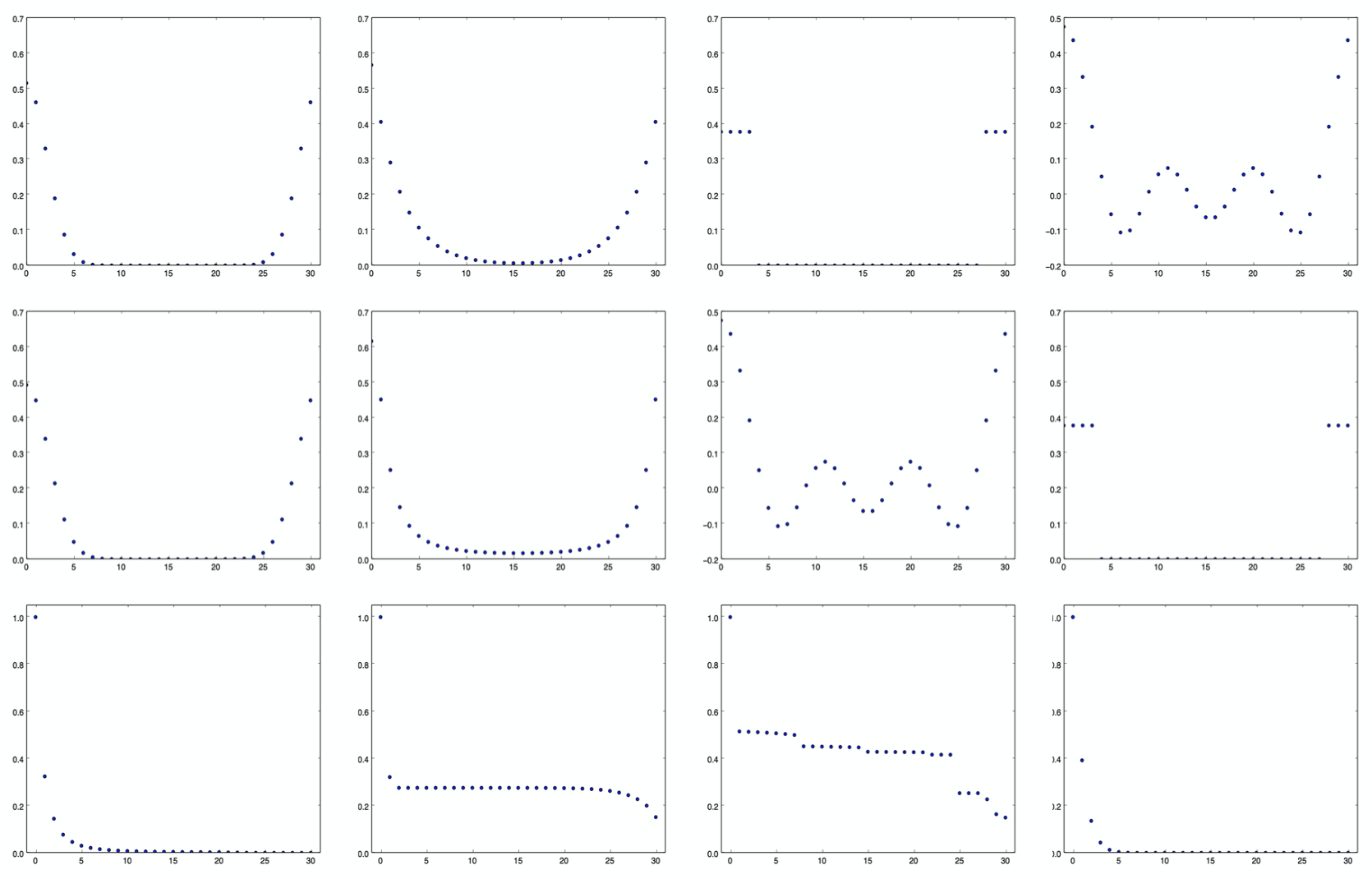}
\caption{Examples of error amplitude $f$ (top), its DFT $\hat f$ (middle), and the length of the $i^{th}$ Gram-Schmidt vector of $M_f$ in Eqn.~\eqref{eqn:M_f} for $0\leq i < q$ (bottom). 
Let $q = 31$ for all examples.
The error amplitude $f$ is (from left to right): (1) Gaussian: $f(x) = \exp(-(x/3)^2)$; (2) Laplacian: $f(x) = \exp(-|x/3|)$; (3) Uniform over $[-3, 3]\cap \Z$; (4) The DFT of Uniform over $[-3, 3]\cap \Z$. }
\end{figure}\label{fig:examplesf}
\else
\begin{figure}[ht]
\centering
\includegraphics[scale=0.20]{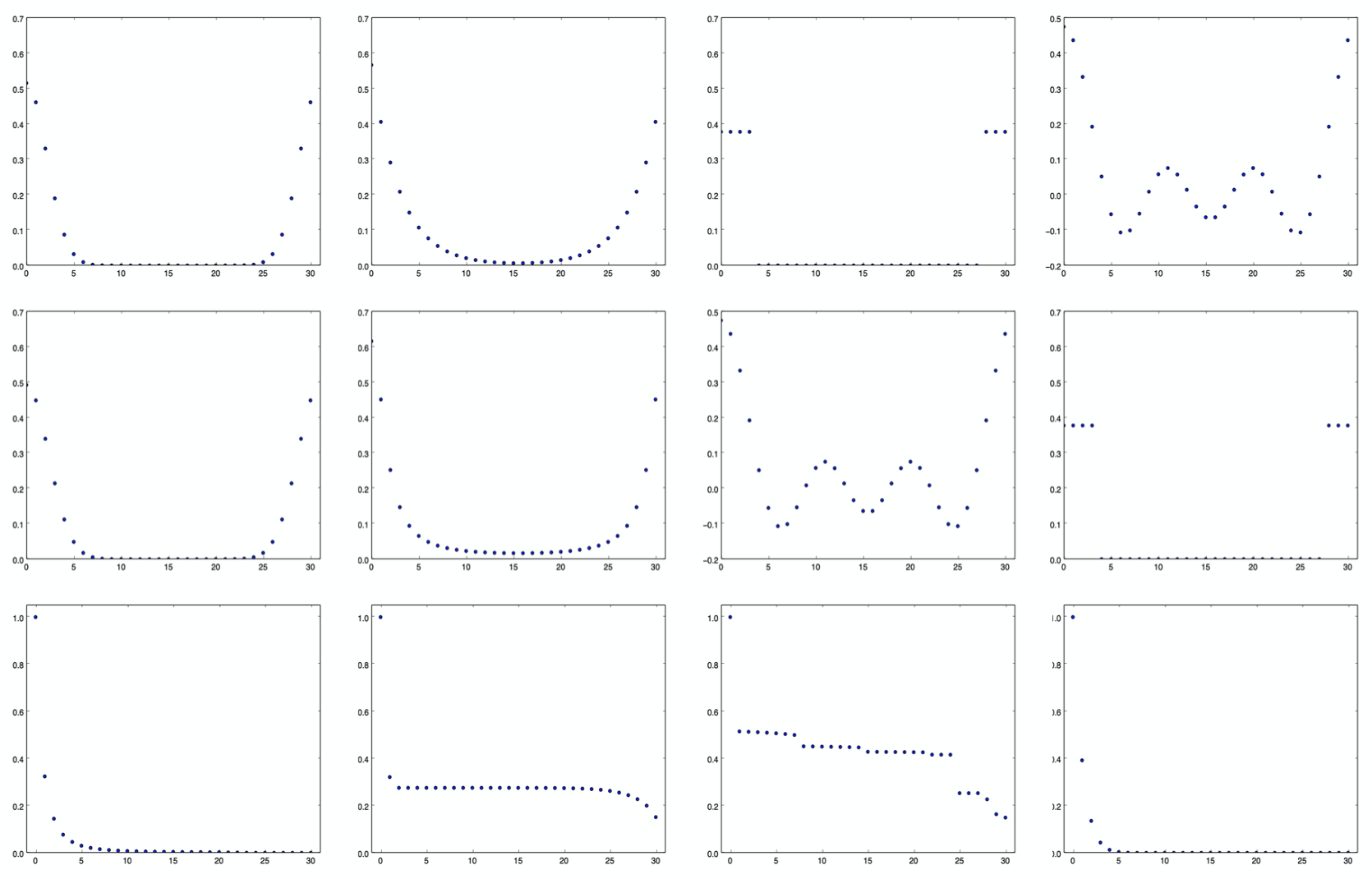}
\caption{Examples of error amplitude $f$ (top), its DFT $\hat f$ (middle), and the length of the $i^{th}$ Gram-Schmidt vector of $M_f$ in Eqn.~\eqref{eqn:M_f} for $0\leq i < q$ (bottom). 
Let $q = 31$ for all examples.
The error amplitude $f$ is (from left to right): (1) Gaussian: $f(x) = \exp(-(x/3)^2)$; (2) Laplacian: $f(x) = \exp(-|x/3|)$; (3) Uniform over $[-3, 3]\cap \Z$; (4) The DFT of Uniform over $[-3, 3]\cap \Z$. }
\end{figure}\label{fig:examplesf}
\fi

\subsection{The related work of Ivanyos et al.}\label{sec:related}
Let us briefly compare our paper with the work of Ivanyos et al.~\cite{DBLP:conf/esa/IvanyosPS18}\footnote{In the initial version of our paper (August 25, 2021) we were not aware of the results in~\cite{DBLP:conf/esa/IvanyosPS18}. We sincerely thank G\'{a}bor Ivanyos for telling us the results in~\cite{DBLP:conf/esa/IvanyosPS18}. }.
As mentioned in \S\ref{sec:EDCPintro}, EDCP with the parameters in Theorem~\ref{thm:DCPwithpolymod_intro} has already been solved in~\cite{DBLP:conf/esa/IvanyosPS18} by a quantum algorithm with similar complexity.
While we solve EDCP using the quantum reduction from EDCP to $\QLWE$ with sinc error distribution, Ivanyos et al. used a reduction from EDCP to a problem called ``learning from disequations'' (LSF), defined as follows: the goal is to learn a secret $s\in Z_q^n$ by querying an oracle which outputs some $a\in \Z_q^n$ such that $\innerprod{a}{s} \in A$, where $A$ is a known subset of $\Z_q$. Given the set $A$ and $m\in n^{O(|A|)}$ samples $a_1, ..., a_m$, they solve LSF in time $n^{O(|A|)}$ classically (using the Arora-Ge algorithm). This means when $|A|$ is a constant, the problem of LSF is solvable in $\poly(n)$ time. 

In their algorithm they also used an idea similar to what we called ``filtering''. While we use filtering to solve $\QLWE$, they used the idea of filtering in the quantum reduction from EDCP to the LSF problem. 

Overall, both papers use the idea of filtering to solve EDCP for the parameters settings in Theorem~\ref{thm:DCPwithpolymod_intro}, but the intermediate problems we reduced to are different. It appears that solving $\QLWE$ allows us to obtain a richer variety of algorithms. In particular, it allows us to obtain a quantum algorithm for $\sistwo$, which was not achieved in~\cite{DBLP:conf/esa/IvanyosPS18}. 
Furthermore, our results give evidence that the $\QLWE$ and $\LWEstate$ problems are quantumly easier to solve than the classical LWE problem, which shows another hope of solving the worst-case lattice problems. Let us elaborate on this point in the next section. 

\subsection{Future directions}\label{sec:intro:open}

Our results show polynomial time quantum algorithms for variants of average-case lattice problems. They do not appear to affect the security of any lattice-based cryptosystems in use. 
One may ask how far are we from solving standard LWE or approximate SVP for all lattices? 
Here we discuss two potential approaches of extending our results towards those ultimate goals. 
Our first observation is that in order to solve standard LWE, ``all'' what we need to do is to solve $\LWEstate_{n,m,q,f}$ with a smaller $m$ than what we have achieved in Theorem~\ref{thm:solvingLWEstate_intro} or Corollary~\ref{thm:solvingLWEstate_Buniform_intro}. For the simplicity of explanation, assume the parameters $\sigma, B, q$ satisfy $\sigma<B\ll q\in \poly(n)$. To solve decisional $\LWE_{n,m,q,D}$ where the noise distribution $D$ is uniform over $[-\sigma, \sigma]\cap\Z$, it suffices to solve $\LWEstate_{n,m,q,f}$ where $f$ is the uniform distribution over $[-B, B]\cap\Z$, and with $m\leq B/\sigma$. Currently, using our result in Corollary~\ref{thm:solvingLWEstate_Buniform_intro}, we need $m\in \Omega\left( n \cdot q^4 \cdot(2B+1) \right)$, which is polynomial in $n$ but way larger than $B/\sigma$. 

The algorithm of breaking decisional LWE via solving $\LWEstate$ is well-known and was implicitly used in the attempt of designing quantum algorithms for lattice problems in~\cite{eldar2016efficient}. 
Let the decisional $\LWE_{n,m,q,D}$ instance be $(A \in\Z_q^{n\times m}, y\in\Z_q^m)$ where $y$ is either an LWE sample or uniformly random. 
We solve $\LWEstate_{n,m,q,f}$, i.e., construct a state 
\[ \ket{\rho}:= \sum_{u\in\Z_q^n}\bigotimes_{i=1}^m \left(\sum_{e_i\in\Z_q} f(e_i) \ket{ a_i \cdot u + e_i \pmod q}\right).\] 
Let $U_{y}$ denote a unitary operator that maps any $x\in\Z_q^m$ to $x+y$. Then we compute $ \bra{\rho} U_y \ket{\rho}$ by performing a Hadamard test. If $y$ is an LWE sample, we expect the overlap between $\ket{\rho}$ and $\ket{\rho+y}$ to be at least $(1 - \sigma/B)^m$, whereas if $y$ is uniform, we expect the overlap to be $0$. Therefore, if we are able to solve $\LWEstate_{n,m,q,f}$ with $m\leq B/\sigma$, then we can solve decisional $\LWE_{n,m,q,D}$. The distributions $f$ and $D$ in the example can be changed to other ones, but all of them require $m$ to be relatively small in order to break standard LWE.

If we are not able to decrease the number of samples in our solutions of $\LWEstate_{n,m,q,f}$ or $\QLWE_{n,m,q,f}$, another hope of solving worst-case approximate SVP is to modify Regev's reduction~\cite{DBLP:conf/stoc/Regev05}. Recall that Regev reduces worst-case approximate SVP to LWE with Gaussian noise and \emph{arbitrarily polynomially} many classical samples. Suppose we can replace LWE with classical samples by its quantum variants $\LWEstate$ or $\QLWE$, and replace Gaussian distribution by distributions with non-negligible DFT (like bounded uniform or Laplace distributions). Then approximate SVP can be solved using Theorem~\ref{thm:solvingLWEstate_intro} without decreasing the number of samples. However, it is not clear to us whether modifying Regev's reduction is feasible or not. 

%

\subsection{Organization and readers' guide}

The rest of the paper is organized as follows. In \S\ref{sec:prelim} we provide more background of quantum computation and algorithms for lattice problems. In \S\ref{sec:idea_filtering} we provide details of the basic idea of filtering and a mini result for $\sistwo$. The mini result in \S\ref{sec:idea_filtering} will be subsumed by the result in \S\ref{sec:sis_polyq_const_gap}, but the analysis in \S\ref{sec:idea_filtering} is simple and instructive for understanding the main results.
The main results in this paper require some mathematical statements about the Gram-Schmidt orthogonalization of circulant matrices, which will be presented in~\S\ref{sec:GSO_Circulant}. Then we present the quantum algorithms for solving $\QLWE$ and $\LWEstate$ in~\S\ref{sec:filtering_multiple}. The quantum algorithms for $\sistwo$ and $\EDCP$ are given in \S\ref{sec:sis_polyq_const_gap} and \S\ref{sec:EDCP}.

\section{Preliminaries}\label{sec:prelim}

\paragraph{Notation and terminology.} 
Let $\R, \Z, \N$ be the set of real numbers, integers and positive integers. 
For $q\in\N_{\geq 2}$, denote $\Z/q\Z$ by $\Z_q$. 
For $n\in\N$, $[n] := \set{1, ..., n}$. 
When a variable $v$ is drawn uniformly at random from the set $S$, we denote by $v\la U(S)$. 

A vector in $\R^n$ is represented in column form by default. For a vector $v$, the $i^{th}$ component of $v$ will be denoted by $v_i$. %
For a matrix $A$, the $i^{th}$ column vector of $A$ is denoted $a_i$. We use $A^T$ to denote the transpose of $A$, $A^H$ to denote the conjugate transpose of $A$.
The length of a vector is the $\ell_p$-norm $\|v\|_p := (\sum v_i^p)^{1/p}$, or the infinity norm given by its largest entry $\| v \|_{\infty} := \max_i\{|v_i|\}$. 
The length of a matrix is the norm of its longest column: $\|A\|_p := \max_i \|a_i\|_p$. 
By default, we use $\ell_2$-norm unless explicitly mentioned. 

%

\subsection{Quantum background}

We assume the readers are familiar with the basic concepts of quantum computation. All the background we need in this paper is available in standard textbooks of quantum computation, e.g., \cite{DBLP:books/daglib/0046438}. When writing a quantum state such as $\sum_{x \in S } f(x)\ket{x}$, we typically omit the normalization factor except when needed. When a state can be approximately constructed within a negligible distance, we sometimes say the state is constructible and not mention the negligible distance.

\paragraph{Efficiently constructible unitary operators. }
In this paper we will use the fact that all the unitary matrices of polynomial dimension can be efficiently approximated within exponentially small distance. 
\begin{proposition}[Page~191~of~\cite{DBLP:books/daglib/0046438}] 
Any unitary matrix $U$ on an $n$-qubit system can be written as a product of at most $2^{n-1}(2^n-1)$ two-level unitary matrices. 
\end{proposition}
Then, using Solovay-Kitaev Theorem, all the unitary matrices of $\poly(n)$ dimensions (therefore, applied on $O(\log n)$ qubits) can be approximated by $2^{O(\log n)}\in\poly(n)$ elementary quantum gates. 
\begin{proposition}
Let $\mathcal{G}$ denote set of unitary matrices that are universal for two-level gates.
Given a unitary matrix $U\in\C^{d\times d}$, there is a classical algorithm that runs in time $\poly(d)$, outputs a sequence of two-level unitary matrices $U_1, ..., U_m\in \mathcal{G}$ such that $\prod_{i = 1}^m U_i$ approximates $U$ within distance negligible in $d$, and $m\in\poly(d)$.
\end{proposition}

Looking ahead, the quantum algorithms in this work require quantum Fourier transform, superposition evaluations of classical circuits on quantum states and quantum gates that operate on $O(\log n)$ qubits. Thus, all quantum algorithms in the work can be efficiently approximated.

\paragraph{Quantum Fourier Transform.}
For any integer $q\geq 2$, let $\omega_q = e^{2 \pi i / q}$ denote a primitive $q$-th root of unity. 
Define a unitary matrix $F_q\in\C^{q\times q}$ where $(F_q)_{i,j} := \frac{1}{\sqrt{q}}\cdot \omega_n^{ij}$, for $i,j\in\Z_q$. 
\begin{theorem}[QFT]
	The unitary operator $\QFT_q:=F_q$ can be implemented by $\poly(\log q)$ elementary quantum gates. When $\QFT_q$ is applied on a quantum state $\ket{\phi} := \sum_{x\in\Z_q} f(x)\ket{x}$, we have  
	\[\QFT_q \ket{\phi} = \sum_{y\in\Z_q} \hat f(y) \ket{y} := \sum_{y\in\Z_q} \sum_{x\in\Z_q} \frac{1}{\sqrt{q}}\cdot \omega_q^{xy}\cdot f(x)  \ket{y}. \]
\end{theorem}

\subsection{Arora-Ge algorithm for solving LWE}

%

We have defined the SIS, DCP, and LWE problems in the introduction. 
Here let us mention the Arora-Ge algorithm for solving LWE when the support of the error distribution is small.
The following theorem is implicitly proven in~\cite[Section~3]{DBLP:conf/icalp/AroraG11}.
\begin{theorem}\label{thm:AG11}
Let $q$ be a prime, $n$ be an integer. Let $\mathcal{D}_{\sf noise}$ be an error distribution which satisfies: 
\begin{enumerate}
    \item The support of $\mathcal{D}_{\sf noise}$ is of size $D<q$. 
    \item $\Pr[e = 0,\,e \gets \mathcal{D}_{\sf noise}] = \frac{1}{\delta}$ for some $\delta>1$. %
\end{enumerate}
Then, let $N$ be $\binom{n + D}{D}$ and $C$ be a sufficiently large constant. Let $m := C N \delta q \log q $. There is a classical algorithm that solves $\LWE_{n,m,q,\mathcal{D}_{\sf noise}}$ in time $\poly(m)$ and succeeds with probability $1 - q^{-N}$.  %
\end{theorem}
Note that the probability is only taken over the randomness of samples. The algorithm is deterministic.

Suppose the error distribution $\mathcal{D}_{\sf noise}$ is known (which is always the case in our application). We can remove the second condition in Theorem~\ref{thm:AG11} by shifting the error distribution such that the probability of getting $0$ is maximized. More precisely, suppose $\mathcal{D}_{\sf noise}$ outputs some $e'\in\Z_q$ with the highest probability; we can always change an LWE sample $(a_i, y_i)$ to $(a_i, y_i - e')$, and apply Arora-Ge on the shifted samples. 
Thus, we can shift the error distribution so that the probability of getting zero error is at least $1/q$. This transformation gives the following simple corollary.

\begin{corollary}\label{coro:AG11}
Let $q$ be a prime, $n$ be an integer. Let $\mathcal{D}_{\sf noise}$ be an error distribution whose support is of size $D<q$ and known to the algorithm. 
Let $m := C \cdot n^D q^2 \log q $ where $C$ is a sufficiently large constant.
There is a classical algorithm that solves
$\LWE_{n,m,q,\mathcal{D}_{\sf noise}}$ in time $\poly(m)$ and succeeds with probability $1 - O(q^{-n^D})$. 
\end{corollary}

\section{The Idea of Filtering and a Mini Result for \texorpdfstring{$\sistwo$}{SIS-INF}}\label{sec:idea_filtering}

In this section we give more details of the basic idea of \emph{filtering}. Using the basic filtering technique, we obtain a polynomial-time quantum algorithm for $\sistwo_{n,m,q,\beta}$ with $q$ being a constant prime, $m$ being as large as $\Omega(n^{q-1})$, and $1\leq \beta<q/2$. Quantum polynomial-time algorithms for $\sistwo$ with such parameter settings have not been given before. 

\begin{theorem}\label{thm:sisinftyconstantmodulus}
Let $n$ be an integer, $q$ be a constant prime modulus. There is a quantum algorithm running in time $\poly(n)$ that solves $\sistwo_{n,m,q,\beta}$ with $m\in \Omega\left(  \frac{ n^{q-1} q^2 \log q}{0.9 - 1/(2\beta+1)} \right) \subseteq \poly(n)$ and any $\beta\in\Z$ such that $1 \leq \beta < q / 2$. %
\end{theorem}

Note that in the above theorem, $0.9 - 1/(2 \beta + 1)$ is at least $0.56$ for $\beta\geq 1$. Thus, $m$ is in the order of $n^{q-1} q^2 \log q$.

Let us first recall the existing quantum reduction from SIS to the problem of constructing certain LWE states presented implicitly in~\cite{DBLP:conf/asiacrypt/StehleSTX09}, then show the filtering technique and explain how to construct the required LWE states.

\subsection{Recalling the quantum reduction from SIS to LWE}

To give a quantum algorithm for solving $\sistwo_{n,m,q,\beta}$ w.r.t. a uniformly random matrix $A\in\Z_q^{n \times m}$, it suffices to produce a state $\sum_{z\in ([-\beta, \beta]\cap\Z)^m \text{ s.t. } Az = 0 \pmod q} \ket{z}$. As long as the set $([-\beta, \beta]\cap\Z)^m$ contains a non-zero solution for $Az = 0 \pmod q$, we can solve $\sistwo_{n,m,q,\beta}$ with probability $\geq 1/2$ by simply measuring the state.

The following is a quantum reduction from SIS to LWE where the distribution of $z$ is general. Let $f:\Z_q\to \R$ be a function (in the example above, $f$ is the uniform distribution over $[-\beta, \beta]\cap\Z$). We abuse the notation to let $f: \Z_q^m\to R$ be defined as $f(x) = \prod_{i = 1}^m f(x_i)$ (we will clearly state the domain when using $f$). 

\begin{proposition}\label{prop:SIStoLWEquant}
To construct an \emph{SIS state} of the form
\begin{equation*}
\ket{\phi_{\sf SIS}} := \sum_{z\in \Z_q^m \text{ s.t. } Az = 0 \pmod q} f(z) \ket{z}.
\end{equation*}
It suffices to construct an \emph{LWE state} of the following form
\begin{equation*}
\ket{\phi_{\sf LWE}} := \sum_{u\in\Z_q^n} \sum_{e\in\Z_q^m} \hat f(e)  \ket{ u^T A + e^T \pmod q},
\end{equation*}
where $\hat f(e_i) = \sum_{x_i\in\Z_q} \frac{1}{\sqrt{q}}\cdot \omega_q^{e_i x_i}  f(x_i)$, for $i = 1, ..., m$, and $\hat f(e) = \prod_{i = 1, ..., m} \hat f(e_i) = \sum_{x\in\Z_q^m} \frac{1}{\sqrt{q^m}}\cdot\omega_q^{\innerprod{e}{x}} f(x)$. 
\end{proposition}
\begin{proof} $ \QFT_q^m \ket{\phi_{\sf LWE}} = \ket{\phi_{\sf SIS}}. $\end{proof}

The following lemma is immediate from Proposition~\ref{prop:SIStoLWEquant}. 

\begin{lemma}\label{lemma:SIStoLWEstate}
Let $n,m,q$ be any integers such that $m\in\Omega(n\log q)\subseteq \poly(n)$. Let $0<\beta<q/2$. 
Let $f$ be the uniform distribution over $([-\beta, \beta]\cap\Z)^m$. Let $A$ be a %
matrix in $\Z_q^{n\times m}$. 
If there is a polynomial-time quantum algorithm that generates a state negligibly close to $\sum_{u\in\Z_q^n} \sum_{e\in\Z_q^m} \hat f(e)  \ket{ u^T A + e} $, then there is a polynomial-time quantum algorithm that solves $\sistwo_{n,m,q,\beta}$ for $A$.
\end{lemma}

\subsection{Constructing the LWE state via filtering}\label{sec:learn_u}

Now let us describe an algorithm for $\LWEstate$. %
\begin{enumerate}
    \item The algorithm first prepares the following state: 
    \begin{equation*}
        \sum_{x \in \Z_q^m } f(x)\ket{x} \otimes \sum_{u \in \Z_q^n} \ket{u},
    \end{equation*}
    where we assume we work with a function $f$ such that $\sum_{x \in \Z_q^m } f(x)\ket{x}$ can be efficiently generated. If so, then the whole state can be efficiently generated.   
    
    \item It then applies ${\sf QFT}_q^m$ on the $x$ registers and gets: 
    \begin{equation*}
        \left( {\sf QFT}_q^m \sum_{x \in \Z_q^m } f(x)\ket{x} \right) \otimes \sum_{u \in \Z_q^n} \ket{u} = \left( \sum_{e\in\Z_q^m} \hat f(e)  \ket{e}\right) \otimes  \left(\sum_u  \ket{u}\right).
    \end{equation*}
    
    \item It then adds $u^T A$ to the $e$ registers in superposition, the state becomes: 
    \begin{equation}\label{eqn:LWEstate_pre}
          \sum_{u\in\Z_q^n} \sum_{e\in\Z_q^m} \hat f(e)  \ket{ u^T A + e} \otimes \ket{u}. 
    \end{equation}

    \item Suppose there is a quantum algorithm that takes a state $\sum_{u\in\Z_q^n} \sum_{e\in\Z_q^m} \hat f(e) \ket{ u^T A + e} \otimes \ket{u}$, outputs a state that is negligibly close to 
    \begin{equation}\label{eqn:LWEstate}
          \sum_{u\in\Z_q^n} \sum_{e\in\Z_q^m} \hat f(e)  \ket{ u^T A + e} \otimes \ket{0}, 
    \end{equation}
    then we are done.
\end{enumerate}

Let us now explain how to learn the secret $u$ from the following state
\begin{equation}
\label{eq:phiu}
\ket{\phi_u} := \sum_{e\in\Z_q^m} \hat f(e) \ket{ u^T A + e}.
\end{equation}
For convenience, although $\ket{\phi_u}$ depends on $A$, we ignore the subscript since $A$ will be clear from the context. We focus on the case where $q$ is a constant prime, and for $i = 1, ..., m$, $f(e_i) = 1/\sqrt{2\beta+1}$ for $e_i\in[-\beta, \beta]\cap\Z$ and $0$ elsewhere. At the end of this subsection we will prove Theorem~\ref{thm:sisinftyconstantmodulus}. 

For the convenience of the rest of the presentation, let us also define %
\begin{equation}\label{eqn:defpsi_v_mini}
    \text{for }v\in\Z_q\text{,  }\ket{ \psi_v } := \sum_{e\in\Z_q} \hat f(e) \ket{ (v + e) \bmod q }.
\end{equation}
Therefore Eqn~\eqref{eqn:LWEstate_pre} can also be written as 
\begin{equation}\label{eqn:LWEstate_mini}
     \sum_{u\in\Z_q^n} \left( \ket{\phi_u} \otimes \ket{u} \right) = \sum_{u\in\Z_q^n} \left( \bigotimes_{i =1, ..., m}\ket{ \psi_{(u^T A)_i} } \otimes \ket{u} \right).
\end{equation}

Now let us fix a vector $u\in\Z_q^n$. To learn $u$ from $\ket{\phi_u}$, we look at each coordinate $\ket{ \psi_{(u^T A)_i} }$ for $i = 1, ..., m$ separately. 
Let us (classically) pick a uniformly random $y_i\in\Z_q$, then define a $q$-dimensional unitary matrix that always maps $\ket {\psi_{y_i}}$ to $\ket 0$; more precisely,
\begin{equation*}
U_{y_i} := \sum_{j=0}^{q-1} \ket{ j }\bra{ \alpha_{i,j} },
\end{equation*} 
where $\ket{ \alpha_{i,0} } := \ket{ \psi_{y_i} }$ and the rest of the vectors $\{\ket{ \alpha_{i,j} }\}_{j=1}^{q-1}$ are picked arbitrarily as long as $U_{y_i}$ is unitary. 

Looking ahead, we will apply $U_{y_i}$ on $\ket{ \psi_{(u^T A)_i} }$. Suppose we measure $U_{y_i}\ket{ \psi_{(u^T A)_i} }$ and get an outcome in $\set{0, 1, ..., q-1}$. If the outcome is not $0$, then we are 100\% sure that $(u^T A)_i \neq y_i$. This is the basic idea of \emph{filtering}, namely, we will filter out the case where $(u^T A)_i = y_i$ for a randomly chosen $y_i\in\Z_q$. Then we will handle the rest of the $q-1$ possibilities of $(u^T A)_i$ using the Arora-Ge algorithm (recall that we assume $q$ is a constant in this subsection).

To explain why filtering works, consider for any $x, y\in\Z_q$. 
Let $U_{y} := \sum_{j=0}^{q-1} \ket{ j }\bra{ \alpha_{j} }$
where $\ket{ \alpha_{0} } := \ket{ \psi_{y} }$ and the rest of the vectors $\{\ket{ \alpha_{j} }\}_{j=1}^{q-1}$ span the rest of the space which are orthogonal to $\ket {\alpha_0}$. 
Then for any $x$, $\ket{\psi_x}$ can be written as a linear combination of basis $\{\ket {\alpha_j}\}_{j=0}^{q-1}$ (which contains $\ket{\psi_y}$), i.e, 
\begin{equation*}
    \ket{\psi_x} = \sum_{j=0}^{q-1} \bk{ \alpha_{j} }{ \psi_x } \cdot \ket{ \alpha_{j} }.
\end{equation*}

Imagine if we apply $U_{y}$ on $\ket{\psi_x}$ and measure, we will get $q$ different possible outcomes. 
\begin{itemize}
\item  If the outcome is $0$, we know that both $x=y$ and $x\ne y$ can happen. 
        \begin{itemize}
            \item If $x = y$, the outcome is $0$ with probability $1$;
            \item Otherwise, the outcome is $0$ with probability $|\bk{\psi_x}{\psi_y}|^2$, which can still be non-zero.         
        \end{itemize}
\item  If the outcome is not $0$, we know that it can only be the case: $x \ne y$. Because when $x=y$, the measurement will always give outcome $0$. 
\end{itemize}
In the next lemma, we show that if we choose $y$ uniformly at random, the above measurement will give a non-zero outcome with ``good'' probability. 

\begin{lemma}\label{lem:random_measure}
Let $y$ be a uniformly random value in $\Z_q$. Then for any $x\in\Z_q$, the probability of measuring $U_y\ket{\psi_x}$ and getting an outcome not equal to $0$ is at least $1-1/(2\beta+1)$:
\begin{align*}
    \forall x, \Pr_{y \in \Z_q}[s \ne 0 \wedge s \gets M_{\sf st} \circ U_y \ket {\psi_x}] \geq 1 - \frac{1}{2 \beta + 1},
\end{align*}
where $M_{\sf st}$ is the measurement operator in the computational basis. 
\end{lemma}
\begin{proof}
Fixing $y$, the probability of getting outcome $0$ is $|\langle {\psi_y} | {\psi_x} \rangle |^2$. 
The probability of getting a non-zero outcome (when $y$ is chosen uniformly at random) is: $ \frac{1}{q} \sum_{y=0}^{q-1} \left(  1 - |\bk{ \psi_y }{ \psi_x }|^2  \right)$. 

To bound the probability, we define $\ket{ \hat\psi_a } = \sum_{x=0}^{q-1} f(x) \omega_{q}^{- x a} \ket{x}$, we have $\ket{ \psi_a } = \QFT_q\ket{ \hat\psi_a }$. 
For any $x, y$, the inner product $\bk{ \psi_x }{ \psi_y } = \bk{ \hat\psi_x }{ \hat\psi_y }$.  The probability we want to bound is, 
\begin{align*}
  1 -  \frac{1}{q} \sum_{y=0}^{q-1} \bk{ \psi_y }{ \psi_x } \bk{ \psi_x }{ \psi_y }   
=~&  1 -  \frac{1}{q} \sum_{y=0}^{q-1} \text{Tr}\left[  \ket{ \hat\psi_x } \bra{ \hat\psi_x }  \, \ket{ \hat\psi_y } \bra{ \hat\psi_y }  \right]   \\
=~&  1 -  \frac{1}{q} \cdot \text{Tr}\left[  \ket{ \hat\psi_x } \bra{ \hat\psi_x }  \,  \left(\sum_{y=0}^{q-1} \ket{ \hat\psi_y } \bra{ \hat\psi_y }  \right) \right]  .
\end{align*}
Let $\Psi :=\sum_{y=0}^{q-1} \ket{ \hat\psi_y } \bra{ \hat\psi_y }$. It can be simplified as follows: 
\begin{eqnarray*}
\Psi = \sum_{y=0}^{q-1} \ket{ \hat\psi_y } \bra{ \hat\psi_y }  &=& \sum_{y=0}^{q-1} \sum_{x\in\Z_q, x'\in\Z_q} f(x) f(x') \omega_{q}^{(x'- x) y} \ket{x} \bra{x'} \\
    &=&  \sum_{x\in\Z_q, x'\in\Z_q} f(x) f(x')  \sum_{y=0}^{q-1} \omega_{q}^{(x'- x) y} \ket{x} \bra{x'} \\
    &=& q \cdot  \sum_{x=0}^{q-1}  f(x)^2 \ket{x} \bra{x} \\
    &=& \frac{q}{2\beta+1} \sum_{x=-\beta}^{\beta} \ket{x} \bra{x}. 
\end{eqnarray*}
Here, we use the fact that $f(x) = 1/\sqrt{2 \beta + 1}$ for any $x \in [-\beta, \beta] \cap \Z$ and $f(x) = 0$ otherwise.Therefore,
\iffull
\begin{align*}
    1 -  \frac{1}{q} \cdot \text{Tr}\left[  \ket{ \hat\psi_x } \bra{ \hat\psi_x }  \,  \left(\sum_{y=0}^{q-1} \ket{ \hat\psi_y } \bra{ \hat\psi_y }  \right) \right]  
    = 1 - \frac{1}{2\beta+1} \cdot \text{Tr}\left[  \ket{ \hat\psi_x } \bra{ \hat\psi_x } \left(\sum_{x=-\beta}^{\beta} \ket{x} \bra{x}\right) \right] \geq 1 - \frac{1}{2\beta+1},
\end{align*}
where the last inequality follows from $\text{Tr}\left[  \ket{ \hat\psi_x } \bra{ \hat\psi_x } \left(\sum_{x=-\beta}^{\beta} \ket{x} \bra{x}\right) \right]\leq 1$. 
\else
\begin{align*}
    1 -  \frac{1}{q} \text{Tr}\left[  \ket{ \hat\psi_x } \bra{ \hat\psi_x }  \,  \left(\sum_{y=0}^{q-1} \ket{ \hat\psi_y } \bra{ \hat\psi_y }  \right) \right]  
    = 1 - \frac{1}{2\beta+1}  \text{Tr}\left[  \ket{ \hat\psi_x } \bra{ \hat\psi_x } \left(\sum_{x=-\beta}^{\beta} \ket{x} \bra{x}\right) \right],
\end{align*}
which is at least $1 - \frac{1}{2 \beta + 1}$. This follows from $\text{Tr}\left[  \ket{ \hat\psi_x } \bra{ \hat\psi_x } \left(\sum_{x=-\beta}^{\beta} \ket{x} \bra{x}\right) \right]\leq 1$. 
\fi
\end{proof}

Thus, if we measure the superposition $\ket{\phi_u}$ entry-by-entry, with overwhelming probability, we will get at least $(1 - 1/(2\beta+1) - \varepsilon) m$ outcomes which are  not $0$ and at most $(1/(2\beta+1)+\varepsilon) m$ outcomes are $0$ (for any constant $\varepsilon > 0$). Here we choose $\varepsilon = 0.1$.
\begin{lemma}
For any fixed $x_1, \cdots, x_m \in\Z_q$, uniformly random $y_1, \cdots, y_m \in \Z_q$, the probability of measuring $U_{y_i}\ket{\psi_{x_i}}$ for all $i \in [m]$ and at least $(0.9 - 1 / (2\beta + 1)) m$ outcomes being non-zero is at least $1 - O(e^{-m})$. Namely, for any fixed $x_1, \cdots, x_m \in\Z_q$,
\begin{align*} \Pr_{y_1, \cdots, y_m \in \Z_q}\left[ z \geq \left(0.9 - \frac{1}{2 \beta + 1}\right)\cdot m  \,\wedge\, \forall i, s_i \gets M_{\sf st} \circ U_{y_i} \ket {\psi_{x_i}}\right] \geq 1 - O(e^{-m}),
\end{align*}
where $z$ is defined as the number of non-zero outcomes among all $s_1, \cdots, s_m$ and $M_{\sf st}$ is the measurement operator in the computational basis. 
\end{lemma}
\begin{proof}
This is a direct consequence of Lemma~\ref{lem:random_measure} and Chernoff bound. 
\end{proof}

By union bound, it can also be shown that, with probability at least $1 - O(q^n e^{-m})$, when the measurements on each bit are chosen uniformly at random, we will get at least $(0.9-1/(2\beta+1)) \cdot m$ {non-zeros} for all $u \in \Z_q^n$. 
\begin{corollary} \label{lem:random_all_good_measure}
For any fixed $A$ in $\mathbb{Z}_q^{n \times m}$, the probability that for all $u \in\Z_q^n$, measuring $U_{y_i}\ket{\psi_{(u^T A)_i}}$ for all $i \in [m]$ and at least $(0.9 - 1 / (2\beta + 1)) m$ outcomes being non-zero is at least $1 - O(q^n e^{-m})$. Namely, \begin{align*}
    \Pr_{y_1, \cdots, y_m \in \Z_q}\left[\forall u \in \Z_q^n, z_u \geq \left(0.9 - \frac{1}{2 \beta + 1} \right) m \right] \geq 1 - O(q^n e^{-m}), 
\end{align*}
where $z_u$ is defined as the number of non-zero outcomes among all $s_{u,1}, \cdots, s_{u, m}$, each $s_{u, i}$ is defined as the measurement outcome of $U_{y_i} \ket {\psi_{(u^T A)_i}}$.
\end{corollary}
The above corollary implies that, for an overwhelming fraction (at least $1-O(q^n e^{-m})$) of $y_1, \cdots, y_m$, the following event happens with probability at least $1-O(q^n e^{-m})$: for all $u \in \Z_q^n$, measuring $U_{y_i} \ket{\psi_{(u^T A))_i}}$ for all $i\in [m]$ and getting at least $(0.9-1/(2\beta +1) m$ outcomes being non-zero.

We are now ready to state the main theorem.
\begin{theorem}\label{thm:LWEstateuncomputeu}
Let $n$ be an integer, $q$ be a constant prime, $C$ be a sufficiently large constant.
Let $m \geq (0.9 - 1/(2\beta+1))^{-1} \cdot C \cdot n^{q-1} q^2 \log q$. Then there exists a QPT algorithm that with overwhelming probability, given a random $A \in \Z_q^{n \times m}$ and $\sum_{u\in\Z_q^n} \ket{ \phi_{u} } \otimes \ket{u}$,
outputs a state negligibly close to $\sum_{u\in\Z_q^n} \ket{ \phi_{u} }$. Here $\ket{\phi_u}$ is defined in Eqn.~\eqref{eq:phiu}. 
\end{theorem}
\begin{proof}
Our algorithm works as follows on state $\sum_{u \in \Z_q^n} \ket {\phi_u} \ket u = \sum_u \bigotimes_{i=1}^m \ket{\psi_{(u^T A)_i}} \ket u$:
\begin{enumerate}
    \item Pick $m$ uniformly random values $y_1, ..., y_m\in\Z_q$. For each $i\in [m]$, construct a unitary 
    $U_{i} := \sum_{j=0}^{q-1} \ket{ j }\bra{ \alpha_{i,j} }$ where for $j = 0, ..., q-1$, 
    \begin{equation}\label{eqn:U_i_mini}
        \ket{ \alpha_{i,j} } := 
        \begin{cases}
        \ket{ \psi_{y_i} }, & \text{ when } j = 0; \\
        \text{An arbitrary $q$-dim unit vector orthogonal to } \set{ \ket{\alpha_{i,k}} }_{k = 0}^{j-1}, & \text{ for } 1\leq j \leq q-1;
        \end{cases}
    \end{equation}

    \item For $i = 1, ..., m$, apply $U_{i}$ to the $i^{th}$ register, we get 
        \begin{eqnarray*}
            U_i \ket{ \psi_{(u^T A)_i} } &=& U_i \left(  \sum_{j=0}^{q-1} \bk{\alpha_{i,j} }{ \psi_{(u^T A)_i} } \cdot \ket{\alpha_{i,j}} \right)  \\
            &=&  \left(  \sum_{j=0}^{q-1} \bk{\alpha_{i,j} }{ \psi_{(u^T A)_i} } \cdot \ket{ j } \right) =: \sum_{s_{u,i}\in\Z_q} w_{s_{u,i}} \ket{s_{u,i}}.
        \end{eqnarray*}
    Here, $s_{u, i}$ denotes the `measurement outcome' of $U_i \ket {\psi_{(u^T A)_i}}$, but we do not physically measure the register $s_{u, i}$.  We denote the vector $(s_{u, 1}, \cdots, s_{u, m})$ by $s_u$.

    \item Then we apply the quantum unitary implementation of the classical algorithm in~\cite{DBLP:conf/icalp/AroraG11} to $\sum_u \sum_{s_{u}\in\Z_q^m} w_{s_{u}} \ket{s_{u}} \otimes \ket{u} := \sum_u \bigotimes_{i=1}^m \sum_{s_{u,i}\in\Z_q} w_{s_{u,i}} \ket{s_{u,i}} \otimes \ket{u} $. 
    Let the algorithm $D_{y_1, y_2, \cdots, y_m}$ be the following in Fig \ref{alg:D}: 
        \begin{algorithm}
        \caption{Learning $u$ from $u^T A$}\label{alg:D}
        \begin{algorithmic}[1]
        \Procedure{$D_{y_1, y_2, \cdots, y_m}$}{$s_u$}
        \For {each $i = 1 , 2, \cdots, m$}
            \If{ $s_{u, i}\neq 0$ (meaning that $(u^T A)_i\neq y_i$)}
                \State{Let $a_i$ and $y_i \pmod q$ } be a sample of LWE
            \EndIf
        \EndFor
        \State{If there are more than $m' = (0.9 - 1/(2\beta + 1)) m$ samples, it runs  Arora-Ge algorithm over those samples to learn $u$ and outputs $u$.}
        \EndProcedure
        \end{algorithmic}
        \end{algorithm}
        
    For any $y_1, \cdots, y_m \in Z_q$ and $u \in \Z_q^n$, if $s_{u, i} \ne 0$, then $(u^T A)_i \ne y_i$; moreover, the LWE sample $(a_i, y_i)$ has an error distribution with support $\set{1, ..., q-1}$, so Corollary~\ref{coro:AG11} applies here.

    We apply $D_{y_1, y_2, \cdots, y_m}$ in superposition to $\sum_u \sum_{s_{u}\in\Z_q^m} w_{s_{u}} \ket{s_{u}} \otimes \ket{u}$. For every fixed $u \in \Z_q^n$, let ${\sf Bad}_u$ be the set such that if all $s_u \in {\sf Bad}_u$, when we apply this algorithm to $s_u$, it does not compute $u$ correctly.

    By Corollaries \ref{lem:random_all_good_measure} and \ref{coro:AG11}, for an overwhelming fraction ($1-O(q^n e^{-m})$) of $y_1, \cdots, y_m$, for every $u$, $\sum_{s_u \in {\sf Bad}_u} |w_{s_u}|^2 \leq \negl(n)$. This is because:
    \begin{enumerate}
        \item By Corollary \ref{lem:random_all_good_measure}, for an overwhelming fraction of $y_1, \cdots, y_m$, for every $u$, $s_u$ provides at least $(0.9 - 1/(2\beta+1)) m = C \cdot n^{q-1} q^2 \log q$ samples with probability more than $1 - O(q^n e^{-m})$. Since $m \gg n$, it happens with overwhelming probability. 
        
        \item By Corollary \ref{coro:AG11},   as long as there are more than $C\cdot n^{q-1} q^2 \log q$ random samples, Arora-Ge algorithm succeeds with probability more than $1 - O(q^{-n^{q-1}})$. Note that the probability is taken over these random samples; in our case, the probability is taken over $A, y_1, \cdots, y_m$.
    \end{enumerate}
    Thus, for an overwhelming fraction of $A, y_1, \cdots, y_m$, for every $u$, the weight $\sum_{s_u \in {\sf Bad}_u} |w_{s_u}|^2 \leq O(q^n e^{-m} + q^{-n^{q-1}}) = \negl(n)$. 

Therefore, for an overwhelming fraction of $A, y_1, \cdots, y_m$, the resulting state is:

\begin{eqnarray*}
\ket{\phi}&:=& D_{y_1, y_2, \cdots, y_m} \cdot q^{-n/2} \sum_{u \in \mathbb{Z}_q^n} \sum_{s_u \in \mathbb{Z}_q^m} w_{s_u} \ket{ s_u, u }  \\
    &=& q^{-n/2} \sum_{u \in \mathbb{Z}_q^n} \left( \sum_{s_u \not\in {\sf Bad}_u} w_{s_u}  \ket{ s_u, 0 } + \sum_{s_u \in {\sf Bad}_u} w_{s_u}  \ket{ s_u, D_{y_1,\cdots,y_m}(s_u) }  \right)  \\
    &=& q^{-n/2} \sum_{u \in \mathbb{Z}_q^n} \left( \sum_{s_u} w_{s_u}  \ket{ s_u, 0 } + \negl_u(n) \ket{ {\sf err}_u } \right)  \\
    &=& q^{-n/2} \sum_{u \in \mathbb{Z}_q^n} \sum_{s_u} w_{s_u} \ket{s_u, 0}  + \negl(n) \ket{ {\sf err} }.
\end{eqnarray*}
Here $\negl_u(n)$ is a complex number whose norm is negligible in $n$, $\ket {{\sf err}_u}$ is some unit vector. Similarly, it is the case for $\negl(n)$ and $\ket{\sf err}$.

\item Finally, we apply $\bigotimes_{i=1}^m U_{i}^{-1}$ to uncompute the projections and get 
\begin{equation*}
     \bigotimes_{i=1}^m U_{i}^{-1} \ket{ \phi } = \sum_{u\in\Z_q^n} \bigotimes_{i =1, ..., m}\ket{ \psi_{(u^T A)_i} } \otimes \ket{0}  + \negl(n) \ket{{\sf err}'} = \sum_{u\in\Z_q^n} \ket{ \phi_{u} } \otimes \ket{0}  + \negl(n) \ket{{\sf err}'}.
\end{equation*}
\end{enumerate}
Therefore we get a state negligibly close to $\sum_{u\in\Z_q^n} \ket {\phi_u} = \sum_{u\in\Z_q^n} \sum_{e\in\Z_q^m} \hat{f}(e) \ket{ u^T A + e}$. This completes the proof of Theorem~\ref{thm:LWEstateuncomputeu}.
\end{proof}

Finally, by Lemma~\ref{lemma:SIStoLWEstate} and Theorem~\ref{thm:LWEstateuncomputeu}, we complete the proof of Theorem~\ref{thm:sisinftyconstantmodulus}. 

\section{Gram-Schmidt for Circulant Matrices}\label{sec:GSO_Circulant}

The general filtering algorithms used later in this paper construct unitary matrices obtained from applying the normalized Gram-Schmidt orthogonalization (GSO) on circulant matrices. The success probabilities of the general filtering algorithms are related to the norm of the columns in the matrices obtained from GSO. Thus, let us provide some related mathematical background in this section.

Given an ordered set of $k\leq n$ linearly independent vectors $\set{ b_0, ..., b_{k-1} }$ in $\R^n$, 
let $B:= \pmat{ b_0, ..., b_{k-1} } \in \R^{n\times k}$. 
For convenience, we sometimes denote $b_i$ by $B_i$. 
Recall the Gram-Schmidt orthogonalization process.
\begin{definition}[GSO]\label{def:GSO}
The Gram-Schmidt orthogonalization of $B$, denoted as $\gs{B} = \pmat{\gs{b_0}, \cdots, \gs{b_{k-1}}}$, is defined iteratively for $i = 0, ..., k-1$ as
\[ \gs{b_i} = b_i - \sum_{j = 0}^{i-1} \frac{ \innerprod{b_i}{\gs{b_j}} }{ \innerprod{\gs{b_j}}{\gs{b_j}} }\cdot \gs{b_j}. \]
\end{definition}

Let us also define the normalized version of Gram-Schmidt orthogonalization.
\begin{definition}
Given an ordered set of $k\leq n$ linearly independent vectors $\set{ b_0, ..., b_{k-1} }$ in $\R^n$, 
let $B:= \pmat{ b_0, ..., b_{k-1} } \in \R^{n\times k}$.
The normalized Gram-Schmidt orthogonalization of $B$, denoted as $\ngs{B} = \pmat{\ngs{b_0}, \cdots, \ngs{b_{k-1}}}$, is defined for $i = 0, ..., k-1$ as $\ngs{b_i} := \gs{b_i}/\|\gs{b_i}\|_2$ where $\gs{b_i}$ is defined in Definition \ref{def:GSO}. 
\end{definition}

The following lemma is helpful for bounding the length of GSO vectors. %
\begin{lemma}[Derived from Corollary~14 of~\cite{MicciancioLN2012}]\label{lemma:gs_dual}
Let $\mat{D} = \pmat{ d_0, ..., d_{k-1} } := \mat{B}\cdot( \mat{B}^T \cdot \mat{B} )^{-1}$. 
Then we have $\|\gs{b_{k-1}}\|_2 = 1/\| d_{k-1}\|_2$.
\end{lemma}

\paragraph{GSO of circulant matrices.}
Let $C\in \R^{n\times n}$ be a real circulant matrix, defined as 
\begin{equation}\label{eqn:defC}
C:= \pmat{ c_0 & c_1 & c_2 & ... & c_{n-1} \\
	c_{n-1} & c_0 & c_1 & ... & c_{n-2} \\
	c_{n-2} & c_{n-1} & c_0 & ... & c_{n-3} \\
	... & ... & ... & ... & ... \\
	c_{1} & c_{2} & c_3 & ... & c_0 }.
\end{equation}

\begin{fact} \label{fact:circulant_eigen}
The QFT basis is an eigenbasis of a
circulant matrix, namely,
\begin{equation}\label{eqn:eigendecompC}
C = F_n^{-1} \cdot \Lambda \cdot F_n,
\end{equation}
where $(F_n)_{i,j} :=  \frac{1}{\sqrt{n}}\cdot \omega_n^{ij}$, for $0\leq i,j\leq n-1$; $\Lambda := \diag{\lambda_0, ..., \lambda_{n-1}}$, where $\lambda_i := \sum_{j = 0}^{n-1}c_j \cdot \omega_n^{ij}$. In other words, the eigenvalues of $C$ are the QFT of the first row of $C$. %
\end{fact}

In our application, we need to compute the lower bound of the length of the $k^{th}$ column of $\gs{C}$, for some $1\leq k\leq n$ such that the first $k$ columns of $C$ are linearly independent. 
Below we present a lemma for general parameter settings. 
For simplicity, the readers can assume we are interested in the range of parameters where $n$ is a polynomial, and $k$ is either equal to $n$ or $n-c$ where $c$ is a constant. 

\begin{lemma}\label{lemma:GSO_cir_ev}
Let $C = F_n^{-1} \cdot \Lambda \cdot F_n$ be a real circulant matrix where $\Lambda := \diag{\lambda_0, ..., \lambda_{n-1}}$, $\lambda_i := \sum_{j = 0}^{n-1} c_j \cdot \omega_n^{ij}$. 
Suppose $\lambda_0, ..., \lambda_{k-1}$ are non-zero and $\lambda_{k}, ..., \lambda_{n-1}$ are zero. 
Then the length of the $k^{th}$ column of $\gs{C}$, i.e., $\|\gs{C}_{k-1}\|_2$, is lower-bounded by 
\begin{enumerate}
\item If $k=n$, then $\|\gs{C}_{n-1}\|_2 \geq \frac{1}{\sqrt{n}}\cdot \min_{i = 0, ..., n-1} |\lambda_i| $. 

\item If $k<n$, then $\|\gs{C}_{k-1}\|_2 \geq \frac{ \sqrt{n} }{ k \cdot 2^{n-k} }\cdot \min_{i = 0, ..., k-1} |\lambda_i| $. 
\end{enumerate}
\end{lemma}

\iffull

\begin{proof}
If $k=n$, then let $D:= C\cdot( C^T \cdot C )^{-1} = C^{-T} = F_n^{-1} \cdot \Lambda^{-1} \cdot F_n$.  %
Therefore
\begin{equation*}
\begin{split}
\| d_{n-1} \|_2 
&\leq \| d_{n-1} \|_\infty \cdot \sqrt{n} \\ 
&\leq_{(1)} n \cdot \| F_n^{-1} \cdot \Lambda^{-1} \|_{\infty}\cdot \| (F_n)_{n-1} \|_\infty \cdot \sqrt{n} \\ 
&\leq n \cdot \left(\frac{1}{\sqrt{n}}\cdot \max_{i = 0, ..., n-1} |\lambda_i|^{-1} \right) \cdot \frac{1}{\sqrt{n}} \cdot \sqrt{n} \\
&\leq \sqrt{n}\cdot \max_{i = 0, ..., n-1} |\lambda_i|^{-1}.
\end{split}
\end{equation*}

The inequality (1) follows from the fact that $\|A x\|_\infty \leq n \cdot \|A\|_\infty \|x\|_\infty$, where $A \in \C^{n\times n}$ and $x \in \C^n$; and $d_{n-1} = F_n^{-1} \cdot \Lambda^{-1} \cdot (F_n)_{n-1}$.

Then by Lemma~\ref{lemma:gs_dual}, $\|\gs{C}_{n-1}\|_2 \geq \frac{1}{\sqrt{n}}\cdot \min_{i = 0, ..., n-1} |\lambda_i| $.

If $k<n$, let $C^{(k)}\in \R^{n\times k}$ denote the first $k$ columns of $C$. Then $C^{(k)}$ can be written as
\begin{equation}\label{eqn:eigendecompCk}
C^{(k)} = L \cdot \Lambda^{(k)} \cdot R,
\end{equation}
where $L\in \C^{n\times k}$ denotes the first $k$ columns of $F_n^{-1}$; $\Lambda^{(k)} := \diag{\lambda_0, ..., \lambda_{k-1}}\in \R^{k\times k}$; $R\in \C^{k\times k}$ denotes the upper-left block of $F_n$, i.e., $R_{i,j} = \frac{1}{\sqrt{n}}\cdot \omega_n^{ij}$, for $0\leq i,j\leq k-1$.

Let $D:= C^{(k)}\cdot( {C^{(k)}}^T \cdot C^{(k)} )^{-1} = C^{(k)}\cdot( {C^{(k)}}^H \cdot C^{(k)} )^{-1}$ (the second equality uses the property that $C$ is real), 
then 
\begin{equation}\label{eqn:defD}
D = L \cdot \Lambda^{(k)} \cdot R \cdot \left( R^H \cdot {\Lambda^{(k)}}^H \cdot L^H\cdot L \cdot \Lambda^{(k)} \cdot R \right)^{-1}
= L \cdot {\Lambda^{(k)}}^{-H} \cdot R^{-H},
\end{equation}
where we use the property that $L^H\cdot L = I\in\R^{k\times k}$. 

From Lemma~\ref{lemma:gs_dual} we know that $\|\gs{C}_{k-1}\|_2 = 1/\|d_{k-1}\|_2$. To get an lower bound of $\|\gs{C}_{k-1}\|_2$, it suffices to get an upper bound of $\| d_{k-1} \|_2$. To get an upper bound of $\| d_{k-1} \|_2$, we need to get an upper bound of the entries in the $k^{th}$ column of $R^{-H}$, i.e., $\| R^{-H}_{k-1} \|_\infty$. 
To estimate $\| R^{-H}_{k-1} \|_\infty$, we use the fact that $R^H$ is a Vandermonde matrix. 

\begin{proposition}[\cite{rawashdeh2019simple}]\label{prop:van}
Let $V\in \C^{k\times k}$ be a Vandermonde matrix such that $V_{j,\ell} := c_{\ell}^j $ for $0\leq j, \ell\leq k-1$ where $c_0, ..., c_{k-1}$ are distinct complex numbers. 
Then the $k^{th}$ column of $V^{-1}$ is  
$(V^{-1})_{j,k-1} = (-1)^{k-1} \cdot \frac{ 1 }{ \prod_{ 0\leq \ell\leq k-1, m\neq j }{(c_\ell - c_j)} }$.
\end{proposition}

Plug in Proposition~\ref{prop:van} with $c_\ell = \omega_n^{-\ell}$ for $0\leq \ell\leq k-1$, we have
\begin{align}\label{eqn:Rinv1}
(R^{-H})_{j,k-1} = \sqrt{n}\cdot (-1)^{k-1} \cdot \frac{ 1 }{ \prod_{ 0\leq \ell\leq k-1, \ell\neq j }{(\omega_n^{-\ell} - \omega_n^{-j})} }.
\end{align}

Let us now bound the norm of the denominators. For $ j = 0, ..., k-1$:
\begin{equation*}
\begin{split}
 \left| \prod_{ 0\leq \ell\leq k-1, \ell\neq j }(\omega_n^{-\ell} - \omega_n^{-j}) \right|
&=~ \left| \prod_{ 0\leq \ell\leq k-1, \ell\neq j }(1 - \omega_n^{\ell-j}) \right| \\
&=~ \left| \prod_{ 0\leq \ell\leq k-1, \ell\neq j } 2\sin(\pi (\ell-j)/n) \right|  \\
&=~ \left| \prod_{ \ell \in \{n-j, n-j+1, \cdots, n - 1 \} \cap \{1, \cdots, k - 1 - j\} } 2\sin(\pi \ell/n) \right|  \\
&=_{(1)}~ \left|\frac{n}{ 2^{n - k} \cdot \prod_{\ell = k - j}^{n-j-1} \sin(\pi \ell/n) }\right|,
\end{split}
\end{equation*}
where $(1)$ uses the identity $\prod_{\ell = 1}^{n-1} \sin(\ell\pi/n) = n/2^{n-1}$.

Thus, for all $0\leq j\leq k-1$, $| R^{-H}_{j,k-1} |\leq \sqrt{n} \cdot \frac{1}{n/2^{n-k}} = \frac{2^{n-k}}{\sqrt{n}}$. 

Therefore, we have 
\begin{equation*}
\begin{split}
\| d_{k-1} \|_2 
\leq~& k \cdot \| L \cdot {\Lambda^{(k)}}^{-H} \|_{\infty}\cdot \|R^{-H}_{k-1} \|_\infty \cdot \sqrt{n} \\ 
\leq~& k \cdot \left(\frac{1}{\sqrt{n}}\cdot \max_{i = 0, ..., k-1} |\lambda_i|^{-1} \right) \cdot \frac{2^{n-k}}{\sqrt{n}} \cdot \sqrt{n} \\
\leq~& \frac{k \cdot 2^{n-k} }{\sqrt{n}} \max_{i = 0, ..., k-1} |\lambda_i|^{-1}.
\end{split}
\end{equation*}
By Lemma~\ref{lemma:gs_dual}, $\|\gs{C}_{k-1}\|_2 = 1/\|d_{k-1}\|_2 \geq \frac{ \sqrt{n} }{ k \cdot 2^{n-k} }\cdot \min_{i = 0, ..., k-1} |\lambda_i| $. 
\end{proof}

\else 
Please refer to \Cref{proof:gramschmidt} for the full proof. 
\fi

\section{Quantum Algorithm for Solving the LWE State Problems}\label{sec:filtering_multiple}

Recall in our mini result, every time a ``measurement'' (we do not physically implement the measurement) gives a non-zero result; it provides us with an inequality $\innerprod{u}{a_i} \ne y_i$. The algorithm, therefore, collects enough inequalities and then runs Arora-Ge to learn the secret vector $u$. There are two bottlenecks in the previous algorithm: (1) we are only able to filter out one value for $\innerprod{u}{a_i}$; (2) to run Arora-Ge, one needs to collect many samples (up to roughly $n^{q-1}$). Therefore, it is only possible to provide quantum polynomial-time algorithms for $\QLWE$, $\LWEstate$, and $\sistwo$ for a constantly large modulus $q$.

In this section, we generalize the filtering algorithm in a way that allows us to filter out $q-c$ many possible values of $\innerprod{u}{a_i}$ for some constant $c$ even when $q$ is a polynomially large modulus. In the best possible case, the filtering algorithm can filter out $q-1$ possibilities and get the exact value of $\innerprod{u}{a_i}$. Therefore, to learn the secret vector $u\in\Z_q^n$, one can collect roughly $n$ samples and run Gaussian elimination. However, the probability of filtering out $q-1$ or $q-c$ (for some constant $c$) values depends on the concrete $f$ and is typically very small. We will precisely show when such a probability is non-negligible.

We now provide quantum algorithms for $\LWEstate_{n,m,q,f}$ (cf. Def.~\ref{def:LWEstateproblem}) and $\QLWE_{n,m,q,f}$ (cf. Def.~\ref{def:QLWEproblem}). Let us first present the algorithms for a general error amplitude $f$, then state corollaries for some functions $f$ of special interest. Looking ahead, the results in \S\ref{sec:sis_polyq_const_gap} and \S\ref{sec:EDCP} use a slight modification of the algorithms presented in this section. Namely, in this section we will only show algorithms for functions $f$ which allow us to filter out $q-1$ possible values then use Gaussian elimination, whereas the results in \S\ref{sec:sis_polyq_const_gap} and \S\ref{sec:EDCP} require us to deal with a function $f$ that allows us to filter out $q-c$ possible values then use Arora-Ge.

\subsection{Overview of the general filtering algorithm}

Let $q$ be a polynomially large modulus, $f$ be an arbitrary noise amplitude. Define
\begin{equation*}
\forall v\in\Z_q, ~~\ket{\psi_v} := \sum_{e\in\Z_q}  f(e) \ket{ v + e \bmod q}. 
\end{equation*}
Following the basic notations and ideas in \S\ref{sec:learn_u}, let us now explain how to filter out two possible values for $(u^T A)_i$, say we are filtering out $(u^T A)_i = y_i$ and $(u^T A)_i = y_i+1$ where $y_i$ is a random value in $\Z_q$. 
To do so, let us define a basis $\set{ \ket{\alpha_{i,j}} }_{j = 0}^{q-1}$ where 
\[ \ket{ \alpha_{i,0} } = \ket{ \psi_{y_i} }; \ket{ \alpha_{i,1} } = \ngs{ \ket{ \psi_{y_i+1} } }. \]
The rest of the vectors in the basis are picked arbitrarily as long as they are orthogonal to $\ket{ \alpha_{i,0} }$ and $\ket{ \alpha_{i,1} }$.

Define $U_{y_i} := \sum_{j=0}^{q-1} \ket{ j }\bra{ \alpha_{i,j} }$. Suppose we ``measure'' $U_{y_i}\ket{ \psi_{(u^T A)_i} }$ and get an outcome in $\set{0, 1, ..., q-1}$:
\begin{enumerate}
\item If the outcome is $0$, then $(u^T A)_i $ can be any values in $\Z_q$;  
\item If the outcome is $1$, then we are 100\% sure that $(u^T A)_i \neq y_i$, since if $(u^T A)_i = y_i$, then the measurement outcome must be $0$.
\item If the outcome is $\geq 2$, then we are 100\% sure that $(u^T A)_i$ does not equal to $y_i$ or $y_i+1$. 
\end{enumerate}

The idea can be further generalized by continuing to do normalized Gram-Schmidt orthogonalization.
Suppose for a moment that $\ket{\psi_{y_i+j}}$, for $j = 0, ..., q-1$, are linearly independent. 
Then we define unitary matrices
\begin{equation*}
U_{y_i} := \sum_{j=0}^{q-1} \ket{ j }\bra{ \alpha_{i,j} }, \text{ where } \ket{\alpha_{i,j}} = \ngs{ \ket{\psi_{y_i+j}} }.
\end{equation*}
Following the previous logic, if we ``measure'' $U_{y_i}\ket{ \psi_{(u^T A)_i} }$, only the outcome ``$q-1$'' gives us a definitive answer of $(u^T A)_i$, that is, $(u^T A)_i = y_i +q -1 \pmod q$.

\paragraph{The probability of filtering out $q-1$ values.}
It remains to understand the probability of getting the measurement outcome $q-1$. 
\begin{equation*}
\begin{split}
   &\Pr_{y_i\in\Z_q}[ (u^T A)_i = y_i+q -1 \pmod q~\land~ q-1 \leftarrow M_{\sf st}\circ U_{y_i}\ket{ \psi_{(u^T A)_i} } ] \\
=~& \frac{1}{q}\cdot\sum_{j\in\Z_q} |\bk{\alpha_{i,q-1}}{\psi_{y_i+j}}|^2
=\frac{1}{q}\cdot |\bk{\alpha_{i,q-1}}{\psi_{y_i+q-1}}|^2,
\end{split}
\end{equation*}
where the second equality follows from the fact that $\ket{ \alpha_{i,q-1}}$ is defined to be orthogonal to all the states except $\ket{ \psi_{y_i+q-1}}$. Furthermore, 
\begin{equation*}
|\bk{\alpha_{i,q-1}}{\psi_{y_i+q-1}}| = \left|\ngs{\ket{\psi_{y_i+q-1}}}^\dagger \ket{\psi_{y_i+q-1}}\right|
 = \| \gs{\ket{\psi_{y_i+q-1}}} \|_2,
\end{equation*}
i.e., it is exactly the norm of the Gram-Schmidt of $\ket{\psi_{y_i+q-1}}$. This quantity has been shown in~Lemma~\ref{lemma:GSO_cir_ev} to be related to the minimum of $\hat f$ over $\Z_q$, namely,
\begin{equation*}
 \| \gs{\ket{\psi_{y_i+q-1}}} \|_2 \geq \min_{x = 0, ..., q-1} |\hat f(x)|.
\end{equation*}
Therefore, we are able to use the general filtering technique to achieve polynomial-time quantum algorithms for $\QLWE_{n,m,q,f}$ and $\LWEstate_{n,m,q,f}$ where $q$ is polynomially large and $f$ is a function such that the minimum of $\hat f$ over $\Z_q$ is non-negligible.

\subsection{Quantum algorithm for generating LWE states with general error}\label{sec:QAforLWEstate}

\begin{theorem}\label{thm:solvingLWEstate}
Let $q$ be a polynomially large modulus. 
Let $f:\Z_q\to \R$ be the amplitude for the error state such that the state $\sum_{e\in\Z_q} f(e)\ket{e}$ is efficiently constructible and $\eta:=\min_{z\in\Z_q}|\hat{f}(z)|$ is non-negligible.
Let $m\in \Omega\left( n \cdot q / \eta^2 \right)\subseteq \poly(n)$, there exist polynomial-time quantum algorithms that solve $\LWEstate_{n,m,q,f}$ and $\QLWE_{n,m,q,f}$.
\end{theorem}

\begin{proof}
We will describe an algorithm for $\LWEstate_{n,m,q,f}$. The algorithm for $\QLWE_{n,m,q,f}$ appears as a subroutine in the algorithm for $\LWEstate_{n,m,q,f}$.
\begin{enumerate}
	\item The algorithm first prepares the following state: 
	\begin{equation*}
        \bigotimes_{i = 1}^m \left( \sum_{e_i \in \Z_q } f(e_i) \ket{e_i} \right) \otimes \sum_{u \in \Z_q^n} \ket{u}. 
	\end{equation*}
	We abuse the notation of $f$ to let $f(e):= \prod_{i = 1}^{m} f(e_i)$ for $e := (e_1, ..., e_m)$. Then the state above can be written as $\sum_{e\in\Z_q^m} f(e) \ket{e} \otimes \sum_{u\in\Z_q^n} \ket{u}$. 

    \item It then adds $u^T A$ to the $e$ registers in superposition, the state is: 
	\begin{equation}\label{eqn:uA+estates}
          \sum_{u\in\Z_q^n} \sum_{e\in\Z_q^m} f(e) \ket{ u^T A + e} \otimes \ket{u} 
	\end{equation}
    Similarly, let us define 
	\begin{equation}\label{eqn:defpsi_vs}
    \text{for }v\in\Z_q\text{,  }\ket{ \psi_v } := \sum_{e\in\Z_q} f(e) \ket{ (v + e) \bmod q }.
	\end{equation}
	Therefore Eqn~\eqref{eqn:uA+estates} can also be written as 
	\begin{equation}\label{eqn:uA+estate_alts}
          \sum_{u\in\Z_q^n} \bigotimes_{i =1, ..., m}\ket{ \psi_{(u^T A)_i} } \otimes \ket{u}.
	\end{equation}

	\item Pick $m$ uniformly random values $y_1, ..., y_m\in\Z_q$. Construct unitary matrices
	\begin{equation}\label{eqn:U_i_LWE}
	\text{For }1\leq i\leq m, ~~U_{i} := \sum_{j=0}^{q-1} \ket{ j }\bra{ \alpha_{i,j} },\text{ where }\ket{\alpha_{i,j}}:= \ngs{ \ket{ \psi_{y_i+j} } }. %
	\end{equation}

	\item For $i = 1, ..., m$, apply $U_{i}$ to the $i^{th}$ register, we get 
		\begin{eqnarray*}
		    U_i \ket{ \psi_{(u^T A)_i} } &=& U_i \left(  \sum_{j=0}^{q-1} \bk{\alpha_{i,j} }{ \psi_{(u^T A)_i} } \cdot \ket{\alpha_{i,j}} \right)  \\
			&=&  \left(  \sum_{j=0}^{q-1} \bk{\alpha_{i,j} }{ \psi_{(u^T A)_i} } \cdot \ket{ j} \right) =: \sum_{s_{u,i}\in\Z_q} w_{s_{u,i}} \ket{s_{u,i}}.
		\end{eqnarray*}
	
	\item Then we apply the quantum unitary implementation of Gaussian elimination to the superposition $\sum_u \sum_{s_{u}\in\Z_q^m} w_{s_{u}} \ket{s_{u}} := \sum_u \bigotimes_{i=1}^m \sum_{s_{u,i}\in\Z_q} w_{s_{u,i}} \ket{s_{u,i}}$.
	The algorithm $D_{y_1, y_2, \cdots, y_m}$ is described in Algorithm~\ref{alg:GE}. 
		\begin{algorithm}
		\caption{Learning $u$ from $u^T A$}\label{alg:GE}
		\begin{algorithmic}[1]
		\Procedure{$D_{y_1, y_2, \cdots, y_m}$}{$\set{ s_{u,i} }_{1\leq i\leq m}$}
		\For {each $i = 1 , 2, \cdots, m$}
			\If{If $ s_{u,i} = q-1 $ (meaning that $(u^T A)_i = y_i+q-1  \pmod q$)}
				\State{Let $a_i$ and $y_i-1 \bmod q$ } be one sample of the linear system
			\EndIf
		\EndFor
		\State{With overwhelming probability, there are $\geq 2 \cdot n$ random samples (to make sure the linear system is full rank)}
		\State{Run the Gaussian elimination algorithm to learn $u$ and return $u$}
		\EndProcedure
		\end{algorithmic}
		\end{algorithm}
In Lemma~\ref{lemma:mGE}, we prove our parameters guarantee that with overwhelming probability,
\begin{enumerate}
\item There exists a set ${\sf Bad}_u$ such that for all $s_u \in {\sf Bad}_u$, when we apply $D_{y_1, y_2, \cdots, y_m}$ to $s_u$, it does not compute $u$ correctly; 
\item For an overwhelming choice of $y_1, \cdots, y_m$, for all $u$, $\sum_{s_u \in {\sf Bad}_u} |w_{s_u}|^2 = O(q^n e^{-m} + q^{-n}) = \negl(n)$. Here $q^{-n}$ is the probability that the linear system is not full rank with $2 n$ samples.  %
\end{enumerate}
Therefore, {for an overwhelming fraction of $A, y_1, \cdots, y_m$}, the resulting state is:
\begin{eqnarray*}
\ket{\phi}&:=& q^{-n/2} \cdot D_{y_1, y_2, \cdots, y_m} \sum_{u \in \mathbb{Z}_q^n} \sum_{s_u \in \mathbb{Z}_q^m} w_{s_u} \ket{ s_u, u }  \\
	&=&  q^{-n/2} \sum_{u \in \mathbb{Z}_q^n} \left( \sum_{s_u \not\in {\sf Bad}_u} w_{s_u}  \ket{ s_u, 0 } + \sum_{s_u \in {\sf Bad}_u} w_{s_u}  \ket{ s_u, D_{y_1,\cdots,y_m}(s_u) }  \right)  \\
	&=&  q^{-n/2} \sum_{u \in \mathbb{Z}_q^n} \left( \sum_{s_u} w_{s_u}  \ket{ s_u, 0 } + \negl_u(n) \ket{ {\sf err}_u } \right)  \\
	&=&  q^{-n/2} \sum_{u \in \mathbb{Z}_q^n} \sum_{s_u} w_{s_u} \ket{s_u, 0}  + \negl(n) \ket{ {\sf err} }.
\end{eqnarray*}
Here $\negl_u(n)$ returns a complex number whose norm is negligible in $n$, $\ket {{\sf err}_u}$ is some unit vector. Similarly, it is the case for $\negl(n)$ and $\ket{{\sf err}}$.

\item Finally, we just apply $ \bigotimes_{i=1}^m U_{i}^{-1}$ to uncompute the projections and get 
\begin{equation*}
	 \bigotimes_{i=1}^m U_{i}^{-1} \ket{ \phi } = \sum_{u\in\Z_q^n} \bigotimes_{i =1, ..., m}\ket{ \psi_{(u^T A)_i} } \otimes \ket{0}  + \negl(n) \ket{{\sf err}'}.
\end{equation*}
\end{enumerate}
Thus, with overwhelming probability, we get a state close to $\sum_{u\in\Z_q^n} \sum_{e\in\Z_q^m} f(e) \ket{u^T A + e}$.  It completes the description of our algorithm. 

\paragraph{The analysis. }
Let us begin with an explanation of the properties of the unitary matrices $U_i$ defined in Eqn~\eqref{eqn:U_i_LWE}. Recall from Eqn.~\eqref{eqn:defpsi_vs} that
\begin{equation*}
\ket{ \psi_v } = \sum_{e\in\Z_q} f(e) \ket{ (v+e) \mod q}.
\end{equation*}

Let $W_i := \sum_{j = 0}^{q-1}\ket{j}\bra{\psi_{y_i+j}}$.
In other words, $W^T_i = \begin{pmatrix} \ket{\psi_{y_i}}, \cdots, \ket{\psi_{y_i + q-1}}\end{pmatrix}$.
Then $U_i^T = \ngs{W_i^T}$.  
We would like to show that the length of the GSO of $\ket{\psi_{y_i+q-1}}$, i.e., the length of the last column of $\gs{W_i^T}$, is non-negligible. 
\begin{lemma}
$\|\gs{ \ket{\psi_{y_i+q-1}} }\|_2 \geq \min_{z\in\Z_q} |\hat f(z)| = \eta $.
\end{lemma}
\begin{proof}
Note that $W_i^T$ is a circulant matrix. The eigenvalues of $W_i^T$ are $\set{ \sqrt{q}\cdot \hat{f}(z) }_{z \in\Z_q}$ (see Fact \ref{fact:circulant_eigen}). Therefore, by applying Lemma~\ref{lemma:GSO_cir_ev}, we have 
$\|\gs{ \ket{\psi_{y_i+q-1}} }\|_2 \geq \min_{z\in\Z_q} |\hat{f}(z)| = \eta$.
\end{proof}

Next, we relate the GSO of $\ket{\psi_{y_i+q-1}}$ to the probability of getting desirable samples in Algorithm~\ref{alg:GE}.
\begin{lemma}\label{lemma:GSOandprob}
For any fixed $x_1, \cdots, x_m \in\Z_q$, 
\begin{align*} \Pr_{y_1, \cdots, y_m \in \Z_q}\left[ z \geq \Omega\left( m\cdot (\eta^2/q) \right)  \,\wedge\, \forall i, s_i \gets M_{\sf st} \circ U_{y_i} \ket {\psi_{x_i}}\right] \geq 1 - O(e^{-m}),
\end{align*}
where $z$ is defined as the number of outcomes such that $ s_{i} = q-1 $ among all $s_1, \cdots, s_m$ and $M_{\sf st}$ is a measurement operator in the computational basis. 
\end{lemma}
\begin{proof}
For $i = 1, ..., m$, we have
\begin{eqnarray*}
& & \Pr_{y_i}[ y_i+q-1 = x_i ]\cdot \Pr\left[ s_{i} = q-1  \,\wedge\, s_i \gets M_{\sf st} \circ U_{y_i} \ket {\psi_{x_i}} \mid y_i+q-1 = x_i \right] \\
&=& \frac{1}{q}\cdot |\bk{\alpha_{i,q-1}}{\psi_{y_i+q-1}}|^2 = \frac{1}{q}\cdot \|\gs{ \ket{\psi_{y_i+q-1}} }\|_2^2  \geq \frac{\eta^2}{q}.
\end{eqnarray*}
The lemma then follows Chernoff bound.
\end{proof}

\begin{lemma}\label{lemma:mGE}
When $m\in \Omega\left( n \cdot q / \eta^2 \right)\subseteq \poly(n)$, for an overwhelming fraction of all possible $A, y_1, \cdots, y_m$, we have: for all $u$,
$\sum_{s_u \in {\sf Bad}_u} |w_{s_u}|^2 \leq \negl(n)$.
\end{lemma}
\begin{proof}
It follows from Lemma~\ref{lemma:GSOandprob} that when $m\in \Omega\left( n \cdot q / \eta^2 \right)$, we have $\geq 2\cdot n$ samples where $(u^T A)_i = y_i-1 \pmod q$ with overwhelming probability. Thus, we can use Gaussian elimination to compute $u$. Therefore $\sum_{s_u \in {\sf Bad}_u} |w_{s_u}|^2 \leq \negl(n)$.
\end{proof}
This completes the proof of Theorem~\ref{thm:solvingLWEstate}.
\end{proof}

\subsection{Examples of error distributions of special interest}\label{sec:QAforLWEstate_special}

We give some examples of error amplitude $f$ where $\min_{y\in\Z_q} |\hat{f}(y)|$ is non-negligible and $q$ is polynomially large.
The first example is where $f$ is the bounded uniform distribution. 

\begin{corollary}\label{thm:solvingLWEstate_Buniform}
Let $q$ be a polynomially large modulus. Let $B\in\Z$ such that $0<2B+1<q$ and $\gcd(2B+1, q) = 1$. 
Let $f:\Z_q\to \R$ be $f(x) := 1/\sqrt{2B+1}$ where $x\in [-B, B]\cap\Z$ and $0$ elsewhere.
Let $m\in \Omega\left( n \cdot q^4 \cdot(2B+1) \right)\subseteq \poly(n)$, there exist polynomial-time quantum algorithms that solve $\LWEstate_{n,m,q,f}$ and $\QLWE_{n,m,q,f}$.
\end{corollary}

\begin{proof}
The QFT of $f$ is
\begin{equation}\label{eqn:FTofUB}
\forall y\in\Z_q, ~~\hat f(y) := \sqrt{\frac{ 1 }{q\cdot (2B+1)}}\cdot \sum_{x = -B}^B \omega_q^{xy}
   = \sqrt{\frac{ 1 }{q\cdot (2B+1)}}\cdot \frac{\sin\left( \frac{2\pi}{q} \cdot\frac{2B+1}{2} \cdot y \right)}{\sin\left( \frac{2\pi}{q}\cdot \frac{y}{2} \right)}.  
\end{equation}
Here we use the identity: $1 + 2 \cos x + \cdots + 2 \cos {n x} = \sin{\left((n + \frac{1}{2}) x\right)} / \sin{\left(\frac{x}{2} \right)}$.

Note that when $y= 0$, $\hat f(y) = \sqrt{\frac{ 2B+1 }{q}}$. When $y\in\set{1, ..., q-1}$, the denominator satisfies $0<\sin\left( \frac{2\pi}{q}\cdot \frac{y}{2} \right)\leq 1$; since $\gcd(2B+1, q) = 1$, we have $\frac{(2B+1)y}{q}\notin\Z$ for any $y \in \set{1, ..., q-1}$,  the numerator satisfies $\left| \sin\left( \frac{2\pi}{q} \cdot\frac{2B+1}{2} \cdot y \right) \right|\geq \left| \sin\left( \frac{\pi}{q} \right) \right| > \frac{1}{q}$. 

Therefore $\eta=\min_{y\in\Z_q}|\hat{f}(y)|\geq \sqrt{\frac{ 1 }{q\cdot (2B+1)}}\cdot\frac{ 1 }{q}$.
The corollary follows by plugging $\eta\geq \sqrt{\frac{ 1 }{q\cdot (2B+1)}}\cdot\frac{ 1 }{q}$ in Theorem~\ref{thm:solvingLWEstate}. 
\end{proof}

\begin{remark}
When $\gcd(2B+ 1,q) = v$ for some $v>1$, we have $\frac{(2B+1)y}{q}\in\Z$ for $q/v-1$ values of $y \in \set{1, ..., q-1}$. Therefore $\hat{f}(y)$ defined in Eqn.~\eqref{eqn:FTofUB} is 0 on $q/v-1$ values. It is not clear to us how to extend our algorithm to the case where $\gcd(2B+ 1,q) >1$.
\end{remark}

Other examples of $f$ where $\min_{y\in\Z_q} |\hat{f}(y)|$ is non-negligible and $q$ is polynomially large include Laplace and super-Gaussian functions. Their $q$-DFT is easier to express by first taking the continuous Fourier transform (CFT) of $f$, denoted as $g$, then discretize to obtain the DFT. Namely, for $y\in\Z_q$, $\hat{f}(y) = \frac{ \sum_{z\in y+q\Z} g(z/q) }{ \sum_{z\in \Z} g(z/q) }$. Let $0<B<q/n^c$ for some $c>0$. %
\begin{enumerate}
	\item Laplace: $f(x) = e^{-|x/B|}$, the CFT of $f$ is $g(y) \propto \frac{2}{1+4(\pi B y)^2}$. 
	\item Super-Gaussian: For $0<p< 2$, $f(x) = e^{-|x/B|^p}$, the CFT of $f$ is asymptotic to $g(y) \propto -\frac{ \pi^{-p-\frac{1}{2}} |By|^{-p-1} \Gamma(\frac{p+1}{2}) }{ \Gamma(-\frac{p}{2}) }$ (see, for example, \cite{miller2019kissing}).
\end{enumerate}

\iffull

\section{Solving \texorpdfstring{$\sistwo$}{SIS-INF} with Polynomial Moduli}\label{sec:sis_polyq_const_gap}

Let us now present our quantum algorithm for solving $\sistwo$. 

\begin{theorem}\label{thm:SISwithpolymod}
Let $c>0$ be a constant integer, $q>c$ be a polynomially large prime modulus. Let $m\in \Omega\left( (q-c)^{3} \cdot n^{{c + 1}} \cdot q\cdot \log q \right)\subseteq \poly(n)$, there is a polynomial time quantum algorithm that solves $\sistwo_{n,m,q,\frac{q-c}{2}}$. 
\end{theorem}

The algorithm uses the quantum reduction from SIS to LWE given in Lemma~\ref{lemma:SIStoLWEstate}. To generate the LWE state needed, we slightly modify the algorithm in \S\ref{sec:filtering_multiple} as follows. Let the LWE noise amplitude be the DFT of the bounded uniform distribution over a support of size $q-c$. In the algorithm for $\LWEstate$, we filter out $q-c$ possible values for some constant $c$ and then use Arora-Ge to learn the secret vector. The reason we only filter out $q-c$ values instead of $q-1$ values is explained in the analysis of the algorithm in Lemma~\ref{lemma:rank}.

\begin{theorem}\label{thm:theLWEstatementneeded}
Let $q$ be a polynomially large prime modulus. Let $B\in\Z$ be such that $q-(2B+1) = c$ is a constant. 
Let $f:Z_q\to \R$ be the bounded uniform distribution over $[-B, B]\cap\Z$. 
Let $m\in \Omega\left( (q-c)^{3} \cdot n^{{c + 1}} \cdot q\cdot \log q \right)\subseteq \poly(n)$.
There exist polynomial time quantum algorithms that solve $\LWEstate_{n,m,q,\hat f}$ and $\QLWE_{n,m,q,\hat f}$. 
\end{theorem}
\begin{proof}
Let $B$ be the bound of the infinity norm of the SIS solution such that $q-(2B+1) = c$ is a constant. Let $\cB:= [-B, B]\cap \Z$. Our goal is to generate an LWE state where the error distribution is the quantum Fourier transformation of the $B$-bounded uniform state. 

The algorithm for generating the LWE state is given as follows.
\begin{enumerate}
    \item The algorithm first prepares the following state: 
	\begin{equation*}
        \sum_{x \in \cB^m }  \ket{x} \otimes \sum_{u \in \mathbb{Z}_q^n} \ket{u},
	\end{equation*}
    where we can view the state on $x$ registers as $\sum_{x\in\Z_q^m}  f(x) \ket{x}$, $ f(x) = 1$ if each entry of $x$ is in $\cB^m$ and $ f(x) = 0$ otherwise. This state can be efficiently generated.   
    
    \item It then applies ${\sf QFT}_q^m$ on the $x$ registers and gets: 
	\begin{equation*}
        \left( {\sf QFT}_q^m \sum_{x \in \cB^m}  \ket{x} \right) \otimes \sum_{u \in \mathbb{Z}_q^n} \ket{u} = \left( \sum_{e\in\Z_q^m} \hat f(e)  \ket{e}\right) \otimes  \left(\sum_u  \ket{u}\right). 
	\end{equation*}
    where $\hat f(e) = \frac{1}{\sqrt{q}}\cdot \sum_{x\in\Z_q^m} \omega_q^{e x}  f(x)$. %
    
    \item It then adds $u^T A$ to the $e$ registers in superposition, the state is: 
	\begin{equation}\label{eqn:uA+estate}
          \sum_{u\in\Z_q^n} \sum_{e\in\Z_q^m} \hat f(e)  \ket{ u^T A + e} \otimes \ket{u}. 
	\end{equation}
    Similarly, let us define %
	\begin{equation}\label{eqn:defpsi_v}
    \text{for }v\in\Z_q\text{,  }\ket{ \psi_v } := \sum_{e\in\Z_q} \hat f(e) \ket{ (v + e) \bmod q },
	\end{equation}
	where we abuse the notations of $e$ and $\hat f(e)$ to represent a value and a function on $\Z_q$ instead of $\Z_q^m$. Therefore Eqn~\eqref{eqn:uA+estate} can also be written as 
	\begin{equation}\label{eqn:uA+estate_alt}
          \sum_{u\in\Z_q^n} \bigotimes_{i =1, ..., m}\ket{ \psi_{(u^T A)_i} } \otimes \ket{u}. 
	\end{equation}

	\item Pick $m$ uniformly random values $y_1, ..., y_m\in\Z_q$. For each $i\in 1, ..., m$, construct a unitary $U_{i} := \sum_{j=0}^{q-1} \ket{ j }\bra{ \alpha_{i,j} }$ where for $j = 0, ..., q-1$, 
	\begin{equation}\label{eqn:U_i}
		\ket{ \alpha_{i,j} } := 
		\begin{cases}
		\ngs{ \ket{ \psi_{y_i+j} } }, & \text{ for } 0\leq j \leq 2B; \\
		\text{An arbitrary $q$-dim unit vector orthogonal to } \set{ \ket{\alpha_{i,k}} }_{k = 0}^{j-1}, & \text{ for } 2B+1\leq j \leq q-1;
		\end{cases}
	\end{equation}

	\item For $i = 1, ..., m$, apply $U_{i}$ to the $i^{th}$ register, we get 
		\begin{eqnarray*}
		    U_i \ket{ \psi_{(u^T A)_i} } &=& U_i \left(  \sum_{j=0}^{q-1} \bk{\alpha_{i,j} }{ \psi_{(u^T A)_i} } \cdot \ket{\alpha_{i,j}} \right)  \\
			&=&  \left(  \sum_{j=0}^{q-1} \bk{\alpha_{i,j} }{ \psi_{(u^T A)_i} } \cdot \ket{ j } \right) =: \sum_{s_{u,i}\in\Z_q} w_{s_{u,i}} \ket{s_{u,i}}.
		\end{eqnarray*}
	
	\item Then we apply the quantum unitary implementation of the classical algorithm in~\cite{DBLP:conf/icalp/AroraG11} to the superposition $\sum_u \sum_{s_{u}\in\Z_q^m} w_{s_{u}} \ket{s_{u}} := \sum_u \bigotimes_{i=1}^m \sum_{s_{u,i}\in\Z_q} w_{s_{u,i}} \ket{s_{u,i}}  $. 
	Let the algorithm $D_{y_1, y_2, \cdots, y_m}$ be the following: 
		\begin{algorithm}
		\caption{Learning $u$ from $u^T A$}\label{alg:D3}
		\begin{algorithmic}[1]
		\Procedure{$D_{y_1, y_2, \cdots, y_m}$}{$\set{ s_{u,i} }_{1\leq i\leq m}$}
		\For {each $i = 1 , 2, \cdots, m$}
			\If{If $ s_{u,i} = 2B $ (meaning that $(u^T A)_i\in \set{ y_i+2B, ..., y_i+q-1}$)}
				\State{Let $a_i$ and $y_i$ } be a sample of LWE
			\EndIf
		\EndFor
		\State{Run the Arora-Ge algorithm to learn $u$ and return $u$}
		\EndProcedure
		\end{algorithmic}
		\end{algorithm}

	In Lemma~\ref{lemma:mAG} we prove our parameters guarantee that with overwhelming probability,
\begin{enumerate}
\item There exists a set ${\sf Bad}_u$ such that for all $s_u \in {\sf Bad}_u$, when we apply this algorithm to $s_u$, it does not compute $u$ correctly; 
\item $\sum_{s_u \in {\sf Bad}_u} |w_{s_u}|^2 \leq \negl(n)$, for an overwhelming choice of $A, y_1, \cdots, y_m$. %
\end{enumerate}
Therefore, the resulting state is:
\begin{eqnarray*}
\ket{\phi}&:=& q^{-n/2} \cdot  D_{y_1, y_2, \cdots, y_m} \sum_{u \in \mathbb{Z}_q^n} \sum_{s_u \in \mathbb{Z}_q^m} w_{s_u} \ket{ s_u, u }  \\
	&=& q^{-n/2} \cdot \sum_{u \in \mathbb{Z}_q^n} \left( \sum_{s_u \not\in {\sf Bad}_u} w_{s_u}  \ket{ s_u, 0 } + \sum_{s_u \in {\sf Bad}_u} w_{s_u}  \ket{ s_u, D_{y_1,\cdots,y_m}(s_u) }  \right)  \\
	&=& q^{-n/2} \cdot \sum_{u \in \mathbb{Z}_q^n} \left( \sum_{s_u} w_{s_u}  \ket{ s_u, 0 } + \negl_u(n) \ket{ {\sf err}_u } \right)  \\
	&=& q^{-n/2} \cdot \sum_{u \in \mathbb{Z}_q^n} \sum_{s_u} w_{s_u} \ket{s_u, 0}  + \negl(n) \ket{ {\sf err} }.
\end{eqnarray*}
Here $\negl_u(n)$ is a complex number whose norm is negligible in $n$, $\ket {{\sf err}_u}$ is some unit vector. Similarly, it is the case for $\negl(n)$ and $\ket{\sf err}$.

\item Finally, we just apply $ \bigotimes_{i=1}^m U_{i}^{-1}$ to uncompute the projections and get 
\begin{equation*}
	 \bigotimes_{i=1}^m U_{i}^{-1} \ket{ \phi } = \sum_{u\in\Z_q^n} \bigotimes_{i =1, ..., m}\ket{ \psi_{(u^T A)_i} } \otimes \ket{0}  + \negl(n) \ket{{\sf err}'}.
\end{equation*}
\end{enumerate}
So with overwhelming probability, we get a state close to $\sum_{u\in\Z_q^n} \sum_{e\in\Z_q^m} \hat f(e) \ket{u^T A + e}$.  It completes the description of our algorithm. 

\paragraph{The analysis. }
Let us begin with an explanation of the properties of the unitary matrices $U_i$ defined in Eqn.~\eqref{eqn:U_i}. Recall from Eqn.~\eqref{eqn:defpsi_v} that
\begin{equation*}
\ket{ \psi_v } = \sum_{e\in\Z_q} \sum_{x = -B}^{B} \sqrt{\frac{1}{q}}\cdot \sqrt{\frac{1}{2B+1}} \cdot \omega_q^{e x}  \ket{ v+e }. 
\end{equation*}

Let $W_i := \sum_{j = 0}^{q-1}\ket{j}\bra{\psi_{y_i+j}}$.
In other words, $W^T_i = \begin{pmatrix} \ket{\psi_{y_i}}, \cdots, \ket{\psi_{y_i + q-1}}\end{pmatrix}$.
Then the first $2B+1$ columns of $U_i^T$ are the same as the first $2B+1$ columns of $\ngs{W_i^T}$. 
Let us first understand why we choose ``$2B+1$'' columns. 
\begin{lemma}\label{lemma:rank}
The rank of $W_i$ is $2B+1$.
\end{lemma}
\begin{proof}
By the definition of $\ket{ \psi_{v} }$, 
\[ \QFT_q^{-1} \ket{ \psi_{v} } = \sum_{x = -B}^{B} \sqrt{\frac{1}{2B+1}} \cdot \omega_q^{-v x} \ket{ x }.\]
Therefore, if we define $Q \in \C^{q\times q}$ as (we think of the indexes of $Q$ as values in $\Z_q$)
\begin{equation*}
Q_{j,\ell} = \begin{cases}\sqrt{\frac{1}{2B+1}} \cdot \omega_q^{j\ell} & \text{ if } y_i-B\leq \ell\leq y_i+B \\
			 0 & \text{ else }
			 \end{cases}.
\end{equation*}
Then $W_i^T = \QFT_q \cdot Q $. Therefore the rank of $W_i$ is $2B+1$.
\end{proof}

This explains why we define $\ket{ \alpha_{i,j} }$ in Eqn~\eqref{eqn:U_i} to be the normalized GSO of the first $2B+1$ vectors in $\set{ \ket{\psi_{y_i+j}} }_{j = 0, ..., q-1}$ plus $q-(2B+1)$ arbitrarily orthogonal vectors - we can only guarantee the first $2B+1$ columns are linearly independent. It also explains why we choose to let ``$s_{u,i} = 2B$'' be the successful condition in Algorithm~\ref{alg:D3} - numbers in $\set{2B+1, ..., q-1}$ will never be the outcome of $U_i\ket{\psi_v}$ for any $v\in\Z_q$ since any $\psi_v$ is in the span of the first $2B+1$ vectors of $U_i$.

Next we show that the length of the GSO of $\ket{\psi_{y_i+2B}}$, i.e., the length of the $(2B+1)^{th}$ column of $\gs{W_i^T}$, is non-negligible. 
\begin{lemma}
$\|\gs{ \ket{\psi_{y_i+2B}} }\|_2 \geq \frac{ q }{ (2B+1)^{1.5} \cdot 2^{q - 2B - 1} } $.
\end{lemma}
\begin{proof}
To bound $ \|\gs{ \ket{\psi_{y_i+2B}} }\|_2 $, we note that $W_i^T$ is a circulant matrix. 
By Fact \ref{fact:circulant_eigen}, the norm of non-zero eigenvalues of $W_i^T$ are all equal to $\sqrt{q}\cdot \frac{1}{\sqrt{2B+1}}$.
Therefore, by applying Lemma~\ref{lemma:GSO_cir_ev} with $k=2B+1$, $\lambda_i = \sqrt{\frac{q}{2B+1}}$, %
for $i = 0, ..., 2B$, we have $\|\gs{ \ket{\psi_{y_i+2B}} }\|_2 \geq \frac{ \sqrt{q} }{ (2B+1) \cdot 2^{q - 2B - 1} }\cdot \sqrt{\frac{q}{2B+1}} = \frac{ q }{ (2B+1)^{1.5} \cdot 2^{q - 2B - 1} }$.
\end{proof}

Next, we relate $\|\gs{ \ket{\psi_{y_i+2B}} }\|_2$ to the probability of getting desirable samples in Algorithm~\ref{alg:D3}.
\begin{lemma}\label{lemma:GSOandprob_sis}
For any fixed $x_1, \cdots, x_m \in\Z_q$, 
\begin{align*} \Pr_{y_1, \cdots, y_m \in \Z_q}\left[ z \geq \Omega\left( m\cdot \frac{ q }{ (2B+1)^{3} \cdot 2^{2c} } \right)  \,\wedge\, \forall i, s_i \gets M_{\sf st} \circ U_{y_i} \ket {\psi_{x_i}}\right] \geq 1 - O(e^{-m}),
\end{align*}
where $z$ is defined as the number of outcomes such that $ s_{i} = 2B $ among all $s_1, \cdots, s_m$ and $M_{\sf st}$ is a measurement operator in the computational basis. 
\end{lemma}
\begin{proof}
For any fixed $\ket{\psi_{x_i}}$ and $y_i$ such that $x_i \in\{ y_i + 2 B, \cdots, y_i + q - 1\}$, we bound the probability that the measurement gives $2 B$. Only in this case, we can get information: namely, $x_i$ is in the set $G = \{y_i + 2 B, \cdots, y_i + q - 1\}$. %

For $i = 1, ..., m$, we have 
\begin{eqnarray*}
& & \Pr_{y_i}[ x_i\in G ]\cdot \Pr\left[ s_{i} = 2B  \,\wedge\, s_i \gets M_{\sf st} \circ U_{y_i} \ket {\psi_{x_i}} \mid x_i\in G \right] \\
&\geq& \Pr_{y_i}[ x_i = y_i+2B ]\cdot \Pr\left[ s_{i} = 2B  \,\wedge\, s_i \gets M_{\sf st} \circ U_{y_i} \ket {\psi_{x_i}} \mid x_i = y_i+2B \right] \\
&=& \frac{1}{q}\cdot |\bk{\alpha_{i,2B}}{\psi_{y_i+2B}}|^2 \\ 
&=& \frac{1}{q}\cdot \|\gs{ \ket{\psi_{y_i+2B}} }\|_2^2  \\
&\geq& \frac{ q }{ (2B+1)^{3} \cdot 2^{2c} }. 
\end{eqnarray*}

The lemma then follows Chernoff bound.
\end{proof}

\begin{lemma}\label{lemma:mAG}
Let $c = q - 2B - 1$ be a constant. When $m\in \Omega\left( (2B+1)^{3} \cdot n^{{c + 1}} \cdot q\cdot \log q \right)\subseteq \poly(n)$, the following holds for an overwhelming fraction of all $A, y_1, \cdots, y_m$: for every $u \in \Z_q^n$, $\sum_{s_u \in {\sf Bad}_u} |w_{s_u}|^2 \leq \negl(n)$. 
\end{lemma}
\begin{proof}
For each $i = 1, ..., m$, when $s_{u,i} = 2B$, it means that $(u^T A)_i\in \set{ y_i+2B, ..., y_i+q-1}$. 
Therefore, by setting $a_i$ and $y_i$ as an LWE sample, we know the error is in the set of $\set{2B, ..., q-1}$ of size ${c + 1}$.

Then, by Corollary~\ref{coro:AG11}, we know that when $m\in \Omega\left( (2B+1)^{3} \cdot n^{{c + 1}} \cdot q\cdot \log q \right)$, the Arora-Ge algorithm has sufficiently many samples for solving LWE with $c$ possible error terms, therefore $\sum_{s_u \in {\sf Bad}_u} |w_{s_u}|^2 \leq \negl(n)$.
\end{proof}
This completes the proof of Theorem~\ref{thm:theLWEstatementneeded}.
\end{proof}

\begin{proof}[Proof of Theorem~\ref{thm:SISwithpolymod}]
The proof of Theorem~\ref{thm:SISwithpolymod} follows Theorem~\ref{thm:theLWEstatementneeded} and the SIS to LWE reduction in Lemma~\ref{lemma:SIStoLWEstate}.
\end{proof}

\section{Solving Variants of Dihedral Coset Problems}\label{sec:EDCP}

We have already provided the background of the extrapolated dihedral coset problem (EDCP) in the introduction. Here let us recall the definition (cf.~Def.~\ref{def:EDCP}).

\begin{definition}[Extrapolated Dihedral Coset Problem]
Let $n\in\N$ be the dimension, $q\geq 2$ be the modulus, and a function $D:\Z_q\to \R$, consists of $m$ input states of the form
\[ \sum_{j\in\Z_q} D(j) \ket{j} \ket{ x+j\cdot s}, \]
where $x \in \Z_q^n$ is arbitrary and $s\in\Z_q^n$ is fixed for all $m$ states. We say that an algorithm solves $\EDCP_{n,m,q,D}$ if it outputs $s$ with probability $\poly(1/(n \log q))$ in time $\poly(n \log q)$.
\end{definition}

We show polynomial time quantum algorithms that solve EDCP with the following parameter settings.

\begin{theorem}\label{thm:DCPwithpolymod_fg}
Let $q$ be a polynomially large modulus. 
Let $f:\Z_q\to \R$ be such that the state $\sum_{e\in\Z_q} f(e)\ket{e}$ is efficiently constructible and $\eta:=\min_{z\in\Z_q}|\hat{f}(z)|$ is non-negligible. 
Let $m\in \Omega\left( n \cdot q / \eta^2 \right)\subseteq \poly(n)$. 
There is a polynomial time quantum algorithm that solves $\EDCP_{n,m,q,\hat f}$
\end{theorem}

\begin{theorem}\label{thm:DCPwithpolymod}
Let $c>0$ be a constant integer, $q>c$ be a polynomially large prime modulus. Let $m\in \Omega\left( (q-c)^{3} \cdot n^{{c + 1}} \cdot q\cdot \log q \right)\in \poly(n)$, there is a polynomial time quantum algorithm that solves $\EDCP_{n,m,q,D}$ where $D$ is the uniform distribution on $[0,q-c)\cap\Z$.
\end{theorem}

We use the quantum reduction from EDCP to LWE of Brakerski et al.~\cite{DBLP:conf/pkc/BrakerskiKSW18}. Let us recall their reduction.
\begin{lemma}\label{lemma:EDCPtoLWE}
Let $n,m,q$ be integers. Let $f: \Z_q \to \R$. If there is a polynomial time quantum algorithm that solves $\QLWE_{n, m, q, \hat{f}}$, then there is a polynomial time quantum algorithm that solves $\EDCP_{n,m,q,f}$.
\end{lemma}

\begin{proof}
The proof is the same as the proof of \cite[Theorem~4]{DBLP:conf/pkc/BrakerskiKSW18} except that we generalize the function $f$ in $\EDCP_{n,m,q,f}$ from discrete Gaussian to a general function. 

Given an instance of $\EDCP_{n,m,q,f}$
\begin{equation}  
	\set{ \sum_{j\in\Z_q} f(j) \ket{j} \ket{ x_i+j\cdot s} }_{i = 1, ..., m} 
\end{equation}  
For each $i = 1, ..., m$, first apply $\QFT_q^n$ over the second register, which gives
\begin{equation}  
	\sum_{a_i\in\Z_q^n} \sum_{j\in\Z_q} \omega_q^{\innerprod{a_i}{x_i+j\cdot s}}  f(j) \ket{j} \ket{ a_i} 
\end{equation}  
Then we measure the second register and omit a phase of $\omega_q^{\innerprod{a_i}{x_i}}$, we have
\begin{equation}  
	a_i\la U(\Z_q^n),  \sum_{j\in\Z_q} \omega_q^{\innerprod{a_i}{j\cdot s}}  f(j) \ket{j},
\end{equation}  
Apply $\QFT_q$ over $\ket{j}$, which gives
\begin{equation}  
	a_i\la U(\Z_q^n),~~ 
	\sum_{z\in\Z_q} \sum_{j\in\Z_q} \omega_q^{j\cdot (\innerprod{a_i}{s}+z)}  f(j) \ket{z}
=	\sum_{e\in\Z_q}   \hat{f}(e) \ket{ e-\innerprod{a_i}{s} },
\end{equation}  
where the equality is obtained by a change of variable $e = \innerprod{a_i}{s}+z \mod q$ and the definition of $\QFT_q$. This completes the proof of the lemma.
\end{proof}

\begin{proof}[Proof of Theorems~\ref{thm:DCPwithpolymod_fg} and~\ref{thm:DCPwithpolymod}]
They are the immediate applications of Lemma~\ref{lemma:EDCPtoLWE} on Theorem~\ref{thm:solvingLWEstate} and Theorem~\ref{thm:theLWEstatementneeded}.
\end{proof}

\fi

\iffull
\section*{Acknowledgement}
We sincerely thank G\'{a}bor Ivanyos for telling us the results in~\cite{DBLP:conf/esa/IvanyosPS18}. We would also like to thank Luowen Qian, L\'{e}o Ducas, and the anonymous reviewers for their helpful comments.
\fi

\bibliography{ref}

\appendix

\ifllncs
\section{Proof of \Cref{lemma:GSO_cir_ev}}
\label{proof:gramschmidt}

\fi 

\newcommand{\iter}[1]{^{(#1)}}

\section{An algorithm for solving SIS with non-trivial \texorpdfstring{$\ell_\infty$-norm}{infinite norm} bounds}\label{sec:SIS_infinite_alg}

We sketch an algorithm that we heard from Regev (personal communication) that solves $\sistwo_{n,m,q,T}$ when the modulus $q$ is a composite number and $m$ is very large. %
The idea of the algorithm is as follows. Let $\mat{A}\in\Z_q^{n\times m}$ denote the public matrix for SIS.
Assume $q = 2^k$. The algorithm starts by finding a combination of column vectors in $\mat{A}$ that is equal to $\ary{0}$ mod $2$. This zeros out the LSB (least significant bits). Then we find a combination of those combinations that makes the second bit zero, etc. Each time the effective width of $\mat{A}$ shrinks by a factor of $n+1$, so $m = (n+1)^{k}$ is needed to get a solution with $\ell_\infty$-norm 1. 

Here is a formal description of the algorithm that generalizes the idea to any composite $q$.

\begin{theorem}\label{thm:sisinfff}
Let $n$ be an integer. Let $q = \prod_{i\in[k]}p_i$ for some $k>1$ and (possibly composite and duplicated) factors $p_i$. Let $m = n^k$, $T = \prod_{i\in[k]} \lowerrounding{p_i/2}$. 
There is a classical algorithm that solves $\sistwo_{n-1, m, q, T}$ in time $\poly(m)$.
\end{theorem}

Note that when all the factors of $q$ are $2$ and $3$, say $q = 2^c$, $m = n^c$ where $c$ is a constant, then $\|\ary{x}\|_\infty = 1$ which is the smallest possible $\ell_\infty$-norm one can get. 

\begin{proof}
The algorithm runs the following procedure recursively for $k$ times. 
Define the initial values as $\mat{A}\iter{1}:=\mat{A}$, $m\iter{1}:=m$, $q\iter{1}:=q$. 
For $1\leq i\leq k$:
\begin{enumerate}
\item Partitions $\mat{A}\iter{i}$ in $m\iter{i}/n$ blocks, each block is an $(n-1)\times n$-dimensional matrix. In other words we let $\mat{A}\iter{i} = [\mat{A}\iter{i}_1, ..., \mat{A}\iter{i}_{m\iter{i}/n}]$. 
\item For $1\leq j\leq m\iter{i}/n$, compute a non-zero vector $\ary{z}_j\in\Z^n$ such that		
        \begin{equation*}
			\mat{A}\iter{i}_j \cdot \ary{z}_j = \ary{0} \pmod{p_i}  \text{    and       } \|\ary{z}_j\|_\infty \leq \lowerrounding{p_i/2}
		\end{equation*}
		Note that such vectors are efficiently computable by solving linear systems over $\Z_{p_i}$.
\item Put $\set{ \ary{z}_1, ..., \ary{z}_{m\iter{i}/n} }$ into a matrix $\mat{Y}_i\in \Z^{m\iter{i}\times (m\iter{i}/n)}$ as follows:
	\begin{equation}
		\mat{Y}_i := \pmat{ \ary{z}_1 &  &   &  \\ & \ary{z}_2 &   &  \\ &   & ... & \\ &  &   & \ary{z}_{m\iter{i}/n}},
	\end{equation}
	where the empty spots are zero.
	Note that $\mat{A}\iter{i} \cdot \mat{Y}_i = \mat{0}_{(n-1)\times (m\iter{i}/n)} \pmod{p_i}$ and $\|\mat{Y}_i\|_\infty \leq \lowerrounding{p_i/2}$. 
\item Let $q\iter{i+1} := q\iter{i}/p_i$, $m\iter{i+1} := m\iter{i}/n$, $\mat{A}\iter{i+1}:= \mat{A}\iter{i} \cdot \mat{Y}_i / p_i \mod{q\iter{i+1}}$, and send the new instance $\mat{A}\iter{i}$, $m\iter{i}, q\iter{i}$ to the next iteration.
\end{enumerate}
After $k$ iterations we let $\ary{y}:= \mat{Y}_1 \cdot ... \cdot \mat{Y}_k \in \Z^{m}$ be the final SIS solution.

Let us first verify that $\mat{A}\ary{y} = \ary{0}\pmod q$. Note that 
\begin{align*}
	\mat{A}\ary{y} = \mat{A} \cdot \mat{Y}_1 \cdot ... \cdot \mat{Y}_k 
	&= p_1\cdot \mat{A}\iter{2} \cdot \mat{Y}_2 \cdot ... \cdot \mat{Y}_k    \\
	&= p_1\cdot p_2\cdot  \mat{A}\iter{3} \cdot \mat{Y}_3 \cdot ... \cdot \mat{Y}_k  \\ 
	& \ \ \vdots\\
	&= p_1\cdot ... \cdot p_{k-1}\cdot  \mat{A}\iter{k} \cdot \mat{Y}_k  \\
	&= p_1\cdot ... \cdot p_{k}\cdot \ary{v} = \ary{0} \pmod q
\end{align*}
where $\ary{v}\in\Z^(n-1)$ is some integer vector.

We now verify that $\| \ary{y} \|_\infty \leq \prod_{i\in[k]} \lowerrounding{p_i/2}$. 
Let $\mat{W}_1:=\mat{Y}_1$, $\mat{W}_{i} := \mat{W}_{i-1} \cdot \mat{Y}_i$, for $2\leq i\leq k$. Then 
\begin{align*}
	\ary{y} = \mat{Y}_1 \cdot \mat{Y}_2 \cdot ... \cdot \mat{Y}_k 
	&= \mat{W}_2 \cdot \mat{Y}_3 \cdot ... \cdot \mat{Y}_k    \\
	& \ \ \vdots\\
	&= \mat{W}_{k-1} \cdot \mat{Y}_k = \mat{W}_{k}
\end{align*}
For $2\leq i\leq k$, observe that if 
\begin{enumerate}
\item each row of $\mat{W}_{i-1}$ has at most 1 non-zero entry;
\item each row of $\mat{Y}_i$ has at most 1 non-zero entry;
\end{enumerate}
then each row of $\mat{W}_{i}$ has at most 1 non-zero entry, and $\|\mat{W}_{i}\|_\infty\leq \|\mat{W}_{i-1}\|_\infty\cdot \|\mat{Y}_{i}\|_\infty$. The proof completes by making the observation above through $i = 2, ..., k$. 
\end{proof}

The algorithm presented above essentially runs Gaussian elimination for smaller moduli recursively and then put the solutions together. If we replace Gaussian elimination by better algorithms for solving $\sistwo$ in the recursive steps (for example, using the quantum algorithms in \Cref{sec:sis_polyq_const_gap}), then \Cref{thm:sisinfff} can be  generalized as follows.

\begin{theorem}\label{thm:sisinfffgen}
Let $q = \prod_{i\in[k]}p_i$ for some $k>1$ and (possibly composite and duplicated) factors $p_i$. Let $m = \prod_{i\in[k]}m_i$ where $m_1, ..., m_k\in\N$ are length parameters. Let $\beta = \prod_{i\in[k]}\beta_i$ where $\beta_1, ..., \beta_k\in\N$ are threshold parameters. If there exist algorithms that solve 
$\sistwo_{n,m_i,p_i,\beta_i}$ in time $poly(m_i)$, for $i = 1, ..., k$, then there is an algorithm that solves $\sistwo_{n,m,q,\beta}$ in time polynomial in $m$. The resulting algorithm is quantum if one of the algorithms for $\sistwo_{n,m_i,p_i,\beta_i}$ is quantum.
\end{theorem}

\end{document}